\documentclass[11pt]{article}
\usepackage[utf8]{inputenc}

\usepackage[square,numbers]{natbib}
\usepackage{algorithm}
\usepackage[noend]{algpseudocode}
\newcommand{\LineComment}[1]{\Statex \textit{#1}}
\usepackage{amsmath}
\usepackage{amsthm}
\usepackage{amssymb}
\usepackage{amsfonts}
\usepackage{hyperref}
\usepackage{xcolor}
\usepackage{appendix}
\usepackage{graphicx}
\usepackage{bbm}
\usepackage{url}
\usepackage{subfiles}
\usepackage{enumerate}
\usepackage{enumitem}
\usepackage{caption}
\usepackage{subcaption}
\usepackage{float}
\usepackage[margin=1in]{geometry}
\newtheorem{theorem}{Theorem}[section]

\newtheorem{lemma}[theorem]{Lemma}
\newtheorem{proposition}[theorem]{Proposition}
\theoremstyle{definition}
\newtheorem{definition}[theorem]{Definition}

\newtheorem{mech}{Mechanism}
\newtheorem{attack}{Attack}

\newcommand{\M}{\mathcal{M}}
\newcommand{\E}{\mathbb{E}}
\newcommand{\mt}{\mathtt}
\newcommand{\github}{\href{https://github.com/cilvento/b2dp}{github repository}}

\DeclareMathOperator*{\argmin}{arg\,min}

\usepackage{listings, listings-rust}
\usepackage{xcolor}

\definecolor{codegreen}{rgb}{0,0.6,0}
\definecolor{codegray}{rgb}{0.5,0.5,0.5}
\definecolor{codepurple}{rgb}{0.58,0,0.82}
\definecolor{backcolour}{rgb}{0.96,0.96,0.96}
 
\lstdefinestyle{mystyle}{
    backgroundcolor=\color{backcolour},   
    commentstyle=\color{codegreen},
    keywordstyle=\color{blue},
    numberstyle=\tiny\color{codegray},
    stringstyle=\color{codepurple},
    basicstyle=\ttfamily\footnotesize,
    breakatwhitespace=false,         
    breaklines=true,                 
    captionpos=b,                    
    keepspaces=true,                 
    numbers=left,                    
    numbersep=5pt,                  
    showspaces=false,                
    showstringspaces=false,
    showtabs=false,                  
    tabsize=2
}
 
\title{Implementing the Exponential Mechanism with Base-2 Differential Privacy}
\author{Christina Ilvento\thanks{Harvard John A Paulson School of Engineering and Applied Science. This work was supported in part by the Sloan Foundation.}}
\date{}

\begin{document}

\maketitle

\begin{abstract}
Despite excellent theoretical support, Differential Privacy (DP) can still be a challenge to implement in practice. In part, this challenge is due to the concerns associated with translating arbitrary- or infinite-precision theoretical mechanisms to the reality of floating point or fixed-precision. Beginning with the troubling result of Mironov demonstrating the security issues of using floating point for implementing the Laplace mechanism, there have been many reasonable questions raised concerning the vulnerabilities of real-world implementations of DP.

In this work, we examine the practicalities of implementing the exponential mechanism of McSherry and Talwar. We demonstrate that naive or malicious implementations can result in catastrophic privacy failures. To address these problems, we show that the mechanism can be implemented \textit{exactly} for a rich set of values of the privacy parameter $\varepsilon$ and utility functions with limited practical overhead in running time and minimal code complexity. 

How do we achieve this result? We employ a simple trick of switching from base $e$ to base $2$, allowing us to perform precise base $2$ arithmetic. A short, precise expression is always available for $\varepsilon$, and the only approximation error we incur is the conversion of the base-2 privacy parameter back to base $e$ for reporting purposes. The core base $2$ arithmetic of the mechanism can be simply and efficiently implemented using open-source high precision arithmetic libraries. Furthermore, the exact nature of the implementation lends itself to simple monitoring of correctness and proofs of privacy.
\end{abstract}

\section{Introduction}
As ever more data is collected about individuals for research or business purposes, ensuring that analyses of these data preserve individuals' privacy becomes increasingly important. Differential Privacy \cite{dwork2006calibrating}, a theoretically rigorous framework for evaluating the privacy properties of mechanisms operating on private data, has emerged as the de facto standard for effective privacy protection. 
Differential privacy (DP) requires that the output of any computations over private data should be roughly indistinguishable from the output of those same computations if a single individual had changed her private data. For example, suppose a private database contains information about smoking behavior and demographic information for the patients in a particular health clinic. A differentially private computation over these data would yield the same insights regardless of whether any \textit{specific} individual smoked. That is, the output of a DP computation over the true dataset $\M(d)$ is roughly indistinguishable from the same computation over an \textit{adjacent} database $\M(d')$, in which one individual's smoking behavior was different. Informally, we say that a mechanism $\M$ is $\varepsilon-$differentially private if for all outputs $c$, $\frac{\Pr[\M(d) =c]}{\Pr[\M(d')=c]} \leq e^{\varepsilon}$.

The benefit of this type of definition is that every individual in the database is guaranteed that their individual smoking behavior is not leaked by the outputs of the computation, but the computation can still gain valuable insights that hold in aggregate, e.g., whether a new smoking cessation program is effective. 
Unlike other proposals, e.g., $k$-anonymity, DP gives a privacy guarantee which is robust to subsequent post-processing and combination or composition with other private (or non-private) mechanisms and makes no assumptions about an adversary's computational powers or side information. This makes DP a powerful tool for building privacy-preserving systems constructed from many separate privacy-preserving components.

Unfortunately, the practicalities of building DP systems or mechanisms can be challenging. Many mechanisms satisfy DP by first computing a statistic of the database $f(d)$, and then adding \textit{noise} scaled to the sensitivity of the statistic, i.e., how much the statistic can change between adjacent databases. 
Unfortunately, such mechanisms can be vulnerable to translation problems between the mechanism specification and the reality of limited precision arithmetic (as well as the complexities of timing and other side-channel problems). In particular, adding noise produced from inexact computations over finite precision values may contain artefacts leaking significant information about the original value  
(see \cite{mironov2012significance,gazeau2016preserving,balcer2017differential}). 

In this work, we examine the problem of implementing the popular and highly flexible exponential mechanism of \cite{mcsherry2007mechanism}. 
Unlike additive noise mechanisms, the exponential mechanism samples from a fixed set of publicly known outcomes $O$, choosing each outcome with probability proportional to a utility $u$ determined by the private data, i.e., $\Pr[\M(d) = o]\propto e^{-\varepsilon u(d,o)}$.  
One might initially think that by revealing a result selected from a public list of possible outcomes (rather than the direct result of floating point computations) the exponential mechanism should be immune to attacks based on floating point artefacts. 
However, we show two new classes of attacks on naive  implementations of the exponential mechanism based on inexact floating-point arithmetic which allow an adversary to distinguish between two adjacent databases with high probability of success. 

Our core conceptual contribution is to 
propose a new view of implementing DP in practice with \textit{exact} calculations by switching from base $e$ to base $2$. We show that this implementation technique gives strong privacy guarantees with limited implementation overhead. Our primary technical contributions are (1) demonstrating how to translate the exponential mechanism into base-2; (2) giving simple methods for expressing a wide range of privacy parameters and utility functions exactly; (3) specifying an exact implementation of the base-2 exponential mechanism; (4) providing 
reference implementations of the base-2 exponential mechanism using the GNU Multiple Precision Arithmetic Library and MPFR in Python (via the $\mathtt{gmpy2}$ interface) and Rust (via the $\mathtt{Rug}$ and $\mathtt{gmp-mpfr-sys}$ interfaces) \cite{Granlund12,fousse2007mpfr,gmpy2,SpiteriRug,SpiteriGMP}.\footnote{Reference implementations can be found in the accompanying github repository at \href{https://github.com/cilvento/b2dp}{\url{https://github.com/cilvento/b2dp}}.} Our reference implementations demonstrate the practical benefits of using exact arithmetic for auditing, monitoring correctness and giving proofs of privacy on the implementation itself. 

\subsection{Differential Privacy Preliminaries}
Differential Privacy \cite{dwork2006calibrating} is a strong definition of privacy which requires that a mechanism operating on private data is stable. 
That is, the output of the mechanism cannot change ``too much'' if a single entry in the database is changed. A key strength of DP is that it doesn't apply to a single mechanism: any randomized mechanism satisfying the stability requirement is differentially private. 

\begin{definition}[Pure Differential Privacy]\label{def:DP}
A randomized mechanism $\M$ is $\varepsilon-$differentially private if for all adjacent databases $d\sim d'$, i.e., databases which differ in a single entry,
\[\Pr[\M(d) \in C] \leq e^{\varepsilon}\Pr[\M(d') \in C]\]
where probability is taken over the randomness of $\M$.
\end{definition}


Suppose we wanted to construct a DP mechanism to release a statistic of the database, $f(d)$. Several core DP mechanisms follow the template of computing $f(d)+$ \textit{noise} scaled based on the sensitivity of $f$. The sensitivity of $f$ captures how much $f$ can change if a single entry of the database is changed, and is defined $\Delta f:= \max_{d\sim d'}|f(d) - f(d')|$. Naturally, computing a counting query, e.g. the number of individuals with a salary greater than $ x$, has lower sensitivity than computing an average of a field, e.g, the average salary of all individuals in the database. This follows from observing that the number of individuals with a salary greater than $x$ can change by only $1$ if a single individual's value is changed. However, the \textit{average} salary could change drastically if a single individual's salary is changed.\footnote{e.g., the CEO of a company may make more than all other employees combined so adding or removing such an individual could as much as double the average.}
The best-known of these ``additive noise'' mechanisms is the Laplace Mechanism.
\begin{mech}[The Laplace Mechanism]\label{mech:laplace}
Given a function $f:D \rightarrow \mathbb{R}$ with sensitivity $\Delta f:= \max_{d\sim d' \in D}|f(d)-f(d')|$, a privacy parameter $\varepsilon$ and a database $d$, release $f(d)+\tau$ where $\tau$ is drawn
from the sensitivity scaled Laplace distribution, given by probability density $\mathsf{Lap(t|\frac{\Delta f}{\varepsilon})}=\frac{\varepsilon}{2\Delta f}e^{-\frac{\varepsilon |t|}{\Delta f }}$.
The Laplace Mechanism is $\varepsilon-$DP.
\end{mech}

\noindent Notice that to implement this mechanism as written, one must sample from the Laplace distribution with arbitrary precision. This is problematic in practice as floating point cannot express every possible value in the distribution. In particular, taking $y = \ln(x)$ for uniformly distributed random floating point numbers $x$ (as one would to sample from the Laplace distribution) will result in \textit{gaps} in the set of values for $y$.\footnote{See section \ref{section:laplace}, \cite{ridout2009generating} and \cite{mironov2012significance} for more discussion of inverse transform sampling and the particulars of sampling from the Laplace distribution.} 
As shown by Mironov \cite{mironov2012significance}, this can be exploited by an adversary to determine $f(d)$ from $f(d)+\tau$ by examining the lowest order bits of the result. 
Roughly speaking, this attack takes advantage of the fact that the supports, i.e., the sets of all possible outcomes, of adjacent databases are different due to the differing artefacts of floating point computations.
Fortunately, this particular attack on the Laplace mechanism (and similar concerns with other releases of noisy statistics under finite precision semantics) can be mitigated by careful implementation. 

We briefly highlight several related works. 
First, Gazeau, Miller and Palamidessi  study the problem of bounding privacy loss due to finite-precision semantics \cite{ gazeau2016preserving}. Their approach is to use fixed-precision computation, and to allow for bounded privacy degradation due to error in scaled noise mechanism (e.g., Laplace). 
Balcer and Vadhan consider the question of efficiently constructing differentially private histograms under finite precision semantics \cite{balcer2017differential}. They exhibit efficient mechanisms with accuracy competitive to the Laplace mechanism for counting queries and histograms.
However, these solutions do not directly address issues faced in implementing the exponential mechanism. 
Third, Burke considers an extension of the attack technique of non-identical supports demonstrated by Mironov to the exponential mechanism for digital goods auctions \cite{burke2017}. This work demonstrates that implementation of the exponential mechanism based on inverse transform sampling can yield differences in the outcome spaces for adjacent databases which can be exploited to violate privacy guarantees. 
Finally, Andrysco, Kohlbrenner, Mowery, Jhala, Lerner and Shacham demonstrate the vulnerabilities of floating point implementations of differentially private mechanisms to timing channels \cite{andrysco2015subnormal}. In this work, we are cognizant of timing channels due to logic variation, but we do not address the timing differences in the floating point operations themselves; this is an area for future work. 

\subsubsection*{The Exponential Mechanism.}
The exponential mechanism, proposed by McSherry and Talwar \cite{mcsherry2007mechanism}, is a general purpose DP mechanism for releasing arbitrary strings, values or other outputs according to their private \textit{utility}. 
For example, suppose a company is trying to set a price for a new piece of software. To set the price, they ask a market research firm to survey potential customers about their preferred price point. The market research company knows that customers will only be honest about their preferred prices if their preferences are kept private, so they decide to provide a differentially private \textit{estimate} of the optimal price. 
Unfortunately, computing the optimal price and then adding noise may result in a price that \textit{zero} buyers are willing to pay. For example, if all buyers are willing to pay at most $\$10$, setting a price of $\$10.01$ will result in zero sales.\footnote{Please see Section 3.4 of \cite{dwork2014algorithmic} for a more complete auction-related example.} Using the Laplace mechanism, there is a 50\% chance the company sets a price no buyers are willing to pay!
The exponential mechanism solves this problem in an elegant way: given a \textit{utility} (potential profit) determined by the private database (the customers' maximum prices) for each potential outcome (each price in a pre-determined range), the mechanism samples a single outcome with probability proportional to its utility. 
In this case, the price of $\$10$ would have the best utility and thus a high probability of selection, whereas any price over $\$10$ would have no utility and therefore low probability of selection. 

\begin{mech}[Exponential Mechanism]\label{mech:exp}
Given a privacy parameter $\varepsilon$, an outcome set $O$ and a utility function $u: D \times O \rightarrow \mathbb{R}$ which maps (database, outcome) pairs to a real-valued utility, the exponential mechanism samples a single element from $O$ based on the probability distribution
\[p(o) := \frac{e^{-\varepsilon u(d, o)}}{\sum_{o \in O} e^{-\varepsilon u(d,o)}}\]
If the sensitivity of $u$, $\Delta u := \max_{d \sim d', o \in O} |u(d,o)-u(d',o)| \leq \alpha$, then the exponential mechanism is $2\alpha\varepsilon$-DP. 
\end{mech}

\noindent By convention, the utility function is negated, and lower utility corresponds to higher probability. 

The exponential mechanism is useful in many settings in which the space of possible outcomes (e.g., integer values in a range, preset error strings, synthetic datasets, etc.) is known a priori or the usefulness of a statistic is highly sensitive to noise (e.g., optimal price) but the sensitivity of utility (e.g., profit) is small. In fact, \textit{any} DP mechanism can be expressed as an instance of the exponential mechanism by choosing the appropriate utility function and outcome space (although it may not be the most computationally efficient method) \cite{mcsherry2007mechanism}. For example, the Laplace mechanism is simply the exponential mechanism with utility function $u(f(d),o)=|f(d)-o|$ for outcome space $O = \mathbb{R}$. 

An initial examination of the exponential mechanism might lead the reader to imagine that it is immune to the floating point issues of the Laplace mechanism, as no floating point calculation is released directly and the calculations of $e^x$ shouldn't have ``too much'' error. However, in the next section we will see that naive implementations of the exponential mechanism are vulnerable to floating-point based attacks.

\section{Subverting the Exponential Mechanism}
To understand how a naive implementation of the exponential mechanism can be subverted, we first recall that a floating point number is represented as $m*2^{\sigma x}$, where $m$, the mantissa, $\sigma$, a sign bit, and $x$, the exponent, are all binary integers. (For the purposes of our examples, we assume 5 bits of precision for the mantissa and a maximum (minimum) exponent of $\pm \ell
$.) Our attacks rely on two properties of floating point arithmetic. First, values which are too small to be expressed in the available bits of precision, i.e., smaller than $2^{-\ell}$, are rounded to zero. So it's possible that the outcome of a multiplication of two positive numbers is zero:
\begin{align}
    00001*2^{-\ell}\times 00001*2^{-\ell} = 0 * 2^{0}
\end{align}
Second, adding two floating point numbers of different magnitudes can result in truncated values.
\begin{align}
\begin{array}{c@{\hspace{5pt}}rcr@{{}{}}r}
  & 11111 & & *2^{+\ell}& \\
+ &  & 00001 &*2^{-\ell}  & \\
\hline
  & 11111 & & *2^{+\ell}&  
\end{array}
\end{align}
\noindent Notice that the larger number $11111*2^{+\ell}$ needs all five bits of precision to express its value, so adding any value smaller than $1*2^{+\ell}$ cannot be accounted for. Thus even though a positive number $(1*2^{-\ell})$ was added, the result is equivalent to the original value.

We now translate these two observations into two attacks on a naive Python implementation of the exponential mechanism, outlined in Figure \ref{fig:naive}. This naive implementation first computes the weight of each element using $\mt{np.exp(x)}$ to calculate $e^x$. It then computes the total weight $T$, and normalized cumulative weights, that is $c_i =\frac{ \sum_{j \leq i}w_j}{T}$. To sample from the exponential mechanism distribution, it samples a random value $\mt{index} \sim [0,1]$, and chooses the first index $i$ with normalized cumulative weight larger than $\mt{index}$. This corresponds to sampling each element with probability $\mt{(c\_weights[i] - c\_weights[i-1])/T} \approx w_i/T$. We assume that all utility functions are pre-screened to ensure that they satisfy sensitivity constraints, and our malicious utility functions will have sensitivity $\leq 1$. In fact, the \textit{range} of utilities for our attack utility functions is $1$, and as such these functions would legitimately pass any audit of utility sensitivity. By convention $[n]$ indicates the set of integers $\{1,2,\ldots,n\}$. For outcome spaces, we index the elements of $O$ as $o_{i \in [|O|]}$. We somewhat abuse notation by using $u(o_{i>1})$ and similar to denote utilities $o_i$ for $i \in [n] \backslash \{1\}$. Expressions $\mathtt{np.exp}$, etc indicate the evaluation of a particular implementation, whereas expressions $e^x$ indicate infinite precision unless otherwise noted. 

\begin{figure}
    \centering
\begin{lstlisting}[language=Python]
# The naive exponential mechanism
import numpy as np
# Inputs:
#  eps: the privacy parameter 
#  u: the utility function
#  O: the set of outputs
# Returns: an element in O
def naive_exp_mech(eps, u, O):
    # compute the weight of each element
    weights = [np.exp(-(eps/2.0)*u(o)) for o in O] 
    T = sum(weights)
    # cumulative weights
    c_weights = [sum(weights[0:i]/T)\ 
                 for i in range(1, len(O)+1)] 
    index = np.random.rand() # uniform sample from [0,1]
    # return element corresponding to the random index
    for i in range(0, len(O)):
        if c_weights[i] >= index:
            return O[i] 
\end{lstlisting}
    \caption{Naive reference implementation of the exponential mechanism in Python.}
    \label{fig:naive}
\end{figure}

\begin{attack}[Zero-rounding]\label{attack:zero}
Suppose the attacker wishes to know if Alice is in the database. The attacker chooses an outcome space $O=[k]$ and a utility function such that if Alice is in the database $u(o_1)=x$ and $u(o_{i>1})=x+1$. If Alice is not in the database $u(o_i)=x$ for all $i \in [k]$. Notice that the sensitivity of $u$ is $1$ regardless of the choice of $x$. The attacker then sets $x$ such that $\mathtt{np.exp(-(eps/2)*x)}>0$ but $\mathtt{np.exp(-(eps/2)*(x+1))}=0$. If Alice is in the database, the only viable outcome is $o_1$, but if she is not in the database, all elements have equal (and positive) weight, so a uniformly random element will be selected. Thus, the attacker has distinguishing probability $\frac{k-1}{k}$ on a single run of the mechanism. 
\end{attack}

To mitigate Attack \ref{attack:zero}, one might suggest that we add an assertion that each weight is positive. However, this fix can also be subverted by taking advantage of truncated addition.

\begin{attack}[Truncated addition]\label{attack:truncated}
The adversary determines two values $x_l$ and $x_s$ and an outcome space $O=[k]$ such that
 \begin{align} \mathtt{np.exp(-(eps/2)*x_l)} = & 
 \mathtt{np.exp(-(eps/2)*x_l)} + \sum_{i \in [k]}\mathtt{np.exp(-(eps/2)*(x_s+1))}
 \label{eq:truncated}
 \end{align} 
and 
\begin{align}
\mathtt{2*np.exp(-(eps/2)*(x_l+1))} \approx\mathtt{np.exp(-(eps/2)*(x_l+1))} + \sum_{i \in [k]}\mathtt{np.exp(-(eps/2)*(x_s))}
\label{eq:nottruncated}
\end{align}
Notice that Equation \ref{eq:truncated} results from truncated addition, as $\mathtt{np.exp(-(eps/2)*(x_s+1))}$ is chosen to be small enough that its addition to $ \mathtt{np.exp(-(eps/2)*x_l)} $ is truncated and effectively a no-op.
On the other hand, $\mathtt{np.exp(-(eps/2)*(x_s))}$ is chosen to be large enough that its sum is not truncated, and $k$ is chosen large enough that the summation is approximately equal to $\mathtt{np.exp(-(eps/2)*(x_l+1))}$.

The attacker then chooses a utility function $u(o_1)=x_l+1$
 and $u(o_{i>1})=x_s$ if Alice is not in the database and $u(o_1)=x_l$ and $u(o_{i>1})=x_s+1$ if Alice is in the database. As before, the sensitivity of $u$ is $1$ regardless of the choice of $x_s$ and $x_l$. Notice that if Alice is in the database, the only viable outcome is $o_1$. This is because the the cumulative weight $\mt{c\_weights[o_1]}$ will be equivalent to $1$, and $o_1$ is the the first index with cumulative weight greater than or equal to any sample in $[0,1]$ that will be encountered (Lines 17-19). However, if Alice is not in the database, then $o_1$ is chosen only about half of the time, as the other $k$ outcomes can still claim the region $\approx ( \frac{1}{2},1]$. 
\end{attack}

The troubling property of Attack \ref{attack:truncated} is that all weights are positive, and the attack can be adapted to ensure that all weights exceed some minimum positive value. (Please see the accompanying 
\github~
for practical demonstrations of these attacks.) 
Indeed, one would need to monitor the outcome of each addition in the mechanism to ensure that values are not excessively truncated to try to catch such subversion. However, such monitoring is non-trivial as the additions are not expected to be exactly correct as $e^x$ cannot be expressed exactly in a finite number of bits.

One approach to mitigate these attacks is to clamp the allowed utilities to a range in which additions of the smallest and largest possible weights can be computed safely. But how do we decide on such a range, and what is the impact of the inexact computations on privacy (given the vagaries of floating point \textit{division} as well as addition and approximate computations of $\mathtt{np.exp}$)?
The key problem with declaring a certain range of utilities to be safe is the difficulty of characterizing the privacy loss due to inexact arithmetic. For example, suppose that a range was chosen which has a maximum error on the order of $2^{-k}$. If each individual weight and the total weights are on the order of $2^0$, this error may be negligible from a privacy perspective, but if each individual weight is on the order of $2^{-k}$, it is not. Thus any general guarantee we might give on the privacy loss of the implementation must be pessimistic.
Furthermore, the more involved and complete the defenses for such attacks become, the more complicated and difficult to audit the code becomes, increasing the likelihood of errors and those errors being missed under audit. 
A second approach is to limit analyses to a set of previously vetted utility functions or gadgets. Unfortunately even simple gadgets with well-understood sensitivity properties, e.g., scalar multiplication or addition, can be used to construct the attacks outlined above. 

Furthermore,  more nuanced applications of the exponential mechanism may explicitly rely on sampling from the exact exponential mechanism distribution. For example, suppose that the outcome set is modified based on private information and the proof of privacy relies on the fact that the distributions of the exponential mechanism over these outcome spaces are close, i.e., the difference in the supports has very low weight. Significant error or clamping of small values may break the proof of privacy, as the difference between the distributions may be increased significantly by increasing the relative weight of low probability values. 
The concept of $\delta-$negligibility introduced by Blocki, Datta and Bonneau in  \cite{blocki2016differentially} is a practical example of a use of the exponential mechanism which relies on the exact decay in probabilities of low utility elements, and the efficient approximate-DP min-cut technique presented by Gupta, Ligett, McSherry, Roth and Talwar in \cite{talwar2009approximation} uses a similar argument.

Unfortunately, these more nuanced applications of the exponential mechanism rule out using noisy max with Gumbel noise (see \cite{adams2013Gumbel}) to implement the exponential mechanism. For example, the method proposed by Blocki et al. for efficient weight computation does not directly compute the utility of any integer partition and instead uses cumulative weights for completions of a prefix, which doesn't translate to noisy max.\footnote{An implementation of the method of Blocki et al. based on the exact implementation given in this work is included in the accompanying \github.} In upcoming work, we also present issues with naive translation of bounded noise implementations for the sparse vector technique, which are also relevant to noisy max with Gumbel noise \cite{ilvento2020}.

\section{Implementing the Exponential Mechanism Exactly}
Given these difficulties, we devote the remainder of this work to describing a simple, exact implementation of the exponential mechanism. 
The primary technical observation motivating our solution is that base-2 exponentiation and arithmetic is much easier than base $e$ for computers. Of course, there is more subtlety to the solution than simply performing calculations in base-2, and there are four technical contributions underpinning our results: 
\begin{enumerate}[leftmargin=*]
    \item Base-2 DP: we define base-2 DP and show the simple relationship between base-2 DP and standard base-$e$ DP. We prove that the base-2 exponential mechanism gives an exact privacy guarantee when implemented correctly, i.e., that there is no privacy penalty due to the specifics of the implementation.
    \item Choosing privacy parameters and handling non-integer utility functions: we show how to choose a rich set of privacy parameters $\eta \propto \varepsilon$ such that $2^{-\eta}$ can be computed exactly with limited precision. We demonstrate how to handle utility functions with arbitrary precision if $2^{-\eta u}$ cannot be computed exactly, e.g., $u=\frac{1}{3}$, via randomized rounding. The key insight is that randomized rounding incurs no privacy penalty from inexact implementation as the proof of privacy follows from worst-case rounding behavior.
    \item Sampling from normalized probabilities without division: we show how to sample from normalized probabilities without using division. This allows us to sample from the exact distribution of the exponential mechanism with limited randomness.
    \item Implementation: we give simple and efficient reference implementations of the base-2 exponential mechanism in Python and Rust. 
\end{enumerate}
\subsection{Base-2 differential privacy}
The motivation for base-2 differential privacy follows from a simple observation: computers are very good at exact binary arithmetic. Instead of computing $e^\varepsilon$, we will compute $2^\eta$ for appropriately chosen $\eta$ to achieve the same result from a privacy perspective while still being able to take advantage of exact floating point arithmetic. 
We now formally define base-2 differential privacy.
\begin{definition}[$|_2$ Differential Privacy]\label{def:b2DP}
A randomized mechanism $\M$ is $\eta|_2-$differentially private if for all adjacent databases $d\sim d'$,
\[\Pr[\M(d) \in C] \leq 2^{\eta} \Pr[\M(d') \in C]\]
where probability is taken over the randomness of $\M$.
\end{definition}

\noindent A very simple change of base proves the relationship between base-2 and base-$e$.
\begin{lemma}\label{lemma:DPtob2DP}
Any mechanism which is $\eta|_2$-differentially private is $\ln(2)\eta-$differentially private. 
\end{lemma}
\begin{proof}
$\Pr[\M(d) \in C] \leq 2^{\eta} \Pr[\M(d') \in C]$;
$2^\eta = e^{\ln(2)\eta}$, thus
$\Pr[\M(d) \in C] \leq e^{\ln(2)\eta} \Pr[\M(d') \in C]$.
\end{proof}
\noindent We can also re-state the exponential mechanism in base-2.
\begin{mech}[$|_2$ Exponential Mechanism] 
Given a privacy parameter $\eta$, an outcome set $O$ and a utility function $u: D \times O \rightarrow \mathbb{R}$ which maps (database, outcome) pairs to a real-valued utility, the $|_2$ exponential mechanism samples a single element from $O$ based on the probability distribution
\[p(o) := \frac{2^{-\eta u(d, o)}}{\sum_{o \in O} 2^{-\eta u(d,o)}}\]
If $\Delta u \leq \alpha$, then the base-2 exponential mechanism is $2\alpha\eta|_2$-DP.  
\end{mech}


If $\eta$ and $u(d,o)$ are integers,  the base-2 exponential mechanism can be implemented exactly in floating point so long as the floating point computations have sufficient bits to store the entirety of the intermediate results.\footnote{This is not entirely trivial for division, but the workaround is  addressed in Section \ref{sec:normalized}.} Thus, if we were willing to use privacy parameters $\varepsilon =2c\ln(2)$ for positive integers $c$, and restricted ourselves to integer utility functions, implementation is as simple as choosing a good arbitrary-precision floating point library. In the next section, we show how to support smaller privacy parameters and use non-integer utility functions.

\subsection{Non-integer privacy parameters and utilities}
We first address selection of privacy parameters assuming integer utilities, and then we extend to non-integer utilities.
Throughout this section, we use 
$\lceil x\rceil$ to indicate rounding $x$ up to the nearest integer, $\lfloor x \rfloor$ to indicate rounding down and $b_x$ to indicate the number of bits required to represent a number $x$ exactly in binary (assuming $x$ can be written exactly in binary).

\subsubsection{Expressive privacy parameters}
Recall that to run the base-2 exponential mechanism, we need to compute values of $2^{-\eta u}$. 
Notice that if $2^{-\eta}$ can be computed exactly, then $(2^{-\eta})^{u}$ can be computed exactly by repeated squaring for integer $u$. 
First, we characterize a large set of $\eta$ for which $2^{-\eta}$ can be computed exactly, and then we show that the number of bits of precision needed to express $2^{-\eta}$ is not too large in this set. These two ingredients are sufficient to show that we can compute the weights for the base-2 exponential mechanism exactly for a wide range of privacy parameters with reasonable precision (for integer utility functions).

Consider privacy parameters of the form 
\[\eta = -z\log_2(\frac{x}{2^y})\] 
for positive integer $x$, $y$ and $z$ such that $\frac{x}{2^y} \leq 1$. Notice that $2^{-\eta}=(\frac{x}{2^y})^z$. Both the integer $x$ and $2^{-y}$ can clearly be written exactly in binary. To compute $(x2^{-y})^z$ for integer $z$ we can perform exponentiation by repeated squaring.
Notice that even for $z=1$, for small values of $y$, we can still achieve a large set of $\varepsilon$. For example, $-\log_2(15/2^4)\approx 0.06$, $-\log_2(31/2^5)\approx 0.03$, etc.\footnote{Balcer and Vadhan use a similar technique in setting a privacy parameter $\varepsilon'$ which is close to their desired parameter $\varepsilon$, but  behaves well in the exponent in order to do efficient inverse transform sampling. Please see Theorem 4.7 and Algorithm 4.8 of \cite{balcer2017differential}.}

The primary consideration for the choice of $x,$ $y,$ and $z$ is the bits of precision required to write $2^{-\eta}$. 
Given a binary number $q$ and an integer $n$, $q^n$ can be written with $\max(1,b_q |n|)$ bits of precision. 
%
%
Thus to write $2^{-\eta}$ we require at most $z(y+b_x)$ bits, and for any integer $u$, computing $(2^{-\eta})^u$ will require at most $\max(1,|u|z(y+b_x))$ bits of precision.\footnote{Any number that can be written as a binary string with $p$ bits of precision requires at most $p$ bits of precision in the mantissa or significant of a floating point number. In many cases, fewer bits are needed, e.g., $5*2^7$ can be written as a binary string $1010000000$ (10 bits) or in floating point as $101*2^{111}$ (3 bits mantissa and 3 bits exponent).} 
As shown above, setting $z=1$ and choosing small $y$ still allows for a wide range of $\eta$, so in practice the magnitude of $u$ is the primary concern.

Thus, the main consideration for controlling precision is the range of utilities. 
Controlling the magnitude of $u$ is straightforward given a pre-determined range of acceptable $u$ and clamping any observed utilities to this range. As long as the range is determined independent of the private database, clamping to the range has no impact on the privacy guarantee. The following simple proposition states that clamped utility functions have sensitivity no larger than their unclamped counterparts:
\begin{proposition}
Given a utility function $u$ such that $\Delta u \leq \alpha$,  $\mathsf{clamp}(u, A, B)$ where  \[\mathsf{clamp}(u,A,B)(x,o):=  \min(\max(A,u(x,o)), B)\] has sensitivity $\Delta \mathsf{clamp}(u,A,B)\leq \Delta u$.
\end{proposition}
\noindent The proof of the proposition follows from observing that clamping values cannot increase the difference in utility of adjacent databases.\footnote{Note that settings in which outcome spaces of the mechanism are not equivalent for adjacent databases, e.g. \cite{blocki2016differentially,talwar2009approximation}, that this clamping argument does not hold, and the full range of $u$ specified by the mechanism must be supported at the cost of increased precision.}


\subsubsection*{Determining the minimum precision}
As we will see in the next section, we need precision sufficient to compute any combination (subset sum) of the weights $2^{-\eta u(d,o)}$ for any set of utilities within the specified clamping bounds in order to sample from the desired distribution.
There are two general approaches for determining the minimum precision: (1) run the mechanism at a given precision, and repeat the mechanism at higher precision if inexact arithmetic is performed or (2) determine the minimum precision before running the mechanism using the publicly specified  bounds on the allowed utility values and outcome space size. The first approach is prone to timing channels, and so we prefer the second approach.

There are two methods for determining minimum precision before executing the mechanism: a worst-case theoretical analysis or worst-case empirical procedure. We first discuss the theoretical worst-case. The following Lemma gives a worst-case bound on the precision required for computing weight combinations.
\begin{lemma}\label{lemma:minprec}
Given a range of utilities $u_{min}$ to $u_{max}$, a maximum number of outcomes $o_{max}$ and a privacy parameter $\eta = -z\log_2(\frac{x}{2^y})$ such that $x,y,z$ are positive integers and $\frac{x}{2^y}\leq 1$, the sum of any subset of at most $o_{max}$ weights of the form $2^{-\eta u}$ computed from integer utilities within the range $[u_{min}, u_{max}]$ requires at most $(\max(1,|u_{min}|) $ $+ \max(1,|u_{max}|))z(y + b_x) + o_{max}$ bits of precision.
\end{lemma}
The proof of Lemma \ref{lemma:minprec} appears in the Appendix. 
The downside of this theoretical worst-case is that it ignores any possible cancellation or efficiency in floating point representation (rather than binary strings). 

The worst-case empirical procedure instead computes every hypothetical worst case, and reports the required precision. More concretely, 
given a bound on the number of outcomes ($o_{max}$) and a range of utilities  ($[u_{min},u_{max}]$) with  maximum weight $w_{max}=2^{-\eta u_{min}}$, it suffices to ensure that we have sufficient precision to calculate (1) each weight independently, i.e., $2^{-\eta u}$ for integer $u \in [u_{min}, u_{max}]$ and (2) the maximum possible sum of the weights plus the weight with highest fractional precision ($w_*$), i.e. $\lceil{o_{max}w_{max}}\rceil + w_{*}$. (Note that $w_*$ is not necessarily the smallest weight.) This follows from observing that the maximum number of bits required for the mantissa will be dictated by the largest possible sum of weights and the highest fractional precision needed to express any individual weight. Thus if the precision is sufficient to express $w_{i}$ and $\lceil{o_{max} w_{max}}\rceil + w_{i}$ for all $i \in [u_{min}, u_{max}]$, then it is sufficient to compute the sum of any valid subset of weights. We implement a simple iterative procedure which attempts to compute this set of sums and increases the working precision if any computations are inexact, terminating once all computations can be performed exactly. 
Algorithm \ref{alg:precision} outlines the procedure in detail.

The benefit of the empirical procedure is that it may result in lower, but still safe, working precision compared with the theoretical worst-case. However, it does come at an increased computational cost, as it iterates over all utilities in the range (potentially) several times. We discuss this trade-off in more detail in Section \ref{sec:implementation}.

\begin{algorithm}[htbp]
  \caption{Minimum precision determination}
  \label{alg:precision}
\begin{algorithmic}[1]
  \LineComment{\textbf{Inputs}: $u_{min}$, the minimum utility, $u_{max}$, the maximum utility, $o_{max}$ the maximum number of outcomes, the privacy parameter $\eta$} 
  \LineComment{\textbf{Outputs}: $p$, a sufficient precision no more than twice the size of the minimum precision to successfully run the base-2 exponential mechanism.}
    \Procedure{ComputePrecision}{$u_{min},u_{max},o_{max},\eta$}
    \State $p \leftarrow 1$
    \While{$\textsc{CheckPrecision}(u_{min},u_{max},o_{max},\eta,p)$ fails}
    \State $p \leftarrow 2p$ \Comment{$2p$ can be changed to a fixed increment if preferred.} 
    \EndWhile
    \State \textbf{return} $p$
    \EndProcedure
    \Function{CheckPrecision}{$u_{min},u_{max},o_{max},base,p$}
    \State Set the precision to $p$
    \State \textbf{return failure} on inexact arithmetic \textbf{for:}
    \State $maxsum \leftarrow \sum_{i \in [o_{max}]}2^{-\eta u_{min}}$
    \For{$u \in [u_{min},u_{max}]$}
    \State $combinedsum \leftarrow 2^{-\eta u} + \lceil{maxsum}\rceil$
    \EndFor
    \EndFunction
  \end{algorithmic}
\end{algorithm}

\subsubsection{Non-integer utilities}
Many applications of the exponential mechanism are amenable to integer utility functions, e.g. converting decimal dollar values to cent values. However, there may be cases where the utility function provided (for example, by a third party) cannot be modified directly, or we wish to simulate a distribution (such as the discrete Laplace distribution). We now consider the case of a utility function which provides arbitrary precision non-integer utilities. 

The simple solution to non-integer utilities is to use randomized rounding, i.e. rounding up or down to the nearest integer with probabilities proportional to how close the value is to each. Our basic strategy is to show that randomized rounding \textit{incurs no privacy penalty from inexact implementation}.

\begin{mech}[Randomized Rounding Exponential Mechanism] \label{mech:rrexpmech}
Given a privacy parameter $\varepsilon$, an outcome set $O$ and a utility function $u: D \times O \rightarrow \mathbb{R}$ which maps (database, outcome) pairs to a real-valued utility, the randomized rounding exponential mechanism first assigns an integer proxy utility 
\[\rho(u(d,o)):=  \begin{cases}
\lfloor{u(d,o)} \rfloor \text{ with probability } |u(d,o)-\lceil{u(d,o)} \rceil|\\
 
\lceil{u(d,o)} \rceil \text{ with probability } |u(d,o)-\lfloor{u(d,o)} \rfloor|\\
\end{cases} \]
and then samples a single element from $O$ based on the probability distribution
\[p(o) := \frac{e^{-\varepsilon \rho(u(d, o))}}{\sum_{o \in O} e^{-\varepsilon \rho(u(d,o))}}\]
\end{mech}

We now consider the privacy of randomized rounding in the exponential mechanism. The argument follows from considering the worst possible set of rounding choices, and arguing that privacy loss is no worse than in that case.
\begin{lemma}[Privacy of arbitrary precision randomized rounding]\label{lemma:randomizedrounding}
Given an implementation of a utility function $u: D \times O \rightarrow \mathbb{R}$ which guarantees that for any pair of adjacent databases $|u(d',o)- u(d,o)| \leq \alpha$ for integer $\alpha$ as implemented\footnote{By ``as implemented'' we mean that the implementation of $u$ has sensitivity $\leq \alpha$. If inexact implementation of $u$ results in increased sensitivity, then this must be taken into account.}, the exponential mechanism with a randomized rounding function of arbitrary precision is $2\alpha \varepsilon-$DP.
\end{lemma}
The proof of Lemma \ref{lemma:randomizedrounding} appears in the appendix.
The primary benefit of Lemma \ref{lemma:randomizedrounding} is that implementing randomized rounding does not require high precision to maintain privacy.
Although randomized rounding does not incur a privacy penalty for low precision implementation, it's not immediately obvious how well the exponential mechanism with randomized rounding approximates the original non-integer utilities.
Randomized rounding is a monotone transformation (see Lemma \ref{lemma:monotone}) and as such the relative ordering of probabilities is preserved, and the following simple bound gives some indication that randomized rounding doesn't change probabilities ``too much'':
\begin{proposition}
Take $p_o$ to be the probability that an outcome $o$ is chosen by the exponential mechanism, and $p_o'$ to be the probability that $o$ is chosen by the randomized rounding exponential mechanism. For any $o$, $p_o e^{-2\varepsilon} \leq p_o' \leq p_o e^{2\varepsilon}$.
\end{proposition}

In Section \ref{section:laplace}, we evaluate the error of randomized rounding empirically, and demonstrate that the bounds are reasonably good for applications like simulating Laplace noise.
If tighter error bounds are needed in a given setting, we derive a tighter (but more complicated) bound on the error compared with the unrounded mechanism in the appendix. 
%
%
Informally, we derive this bound by lower-bounding the probability that the randomized rounding mechanism outputs a particular value $o$. As the mechanism is guaranteed to output one value from $O$, we can determine the maximum error of the lower bound by comparing the sum of the lower-bounds for each value $o \in O$ with 1. We then assume the worst case distribution of this error, namely that it is concentrated on a single value. This yields an upper bound on the probability the mechanism outputs any particular $o$. We can then determine a maximum pointwise error between the rounded and unrounded mechanisms by computing the maximum difference between the bounds and the unrounded mechanism probability. This method is simple to implement for reporting purposes, and we give an example set of bounds in the appendix. 

\subsection{Normalized sampling without division}\label{sec:normalized}
The final ingredient required to implement the exponential mechanism with exact arithmetic in base-2 is to demonstrate how to sample an outcome, i.e., how to normalize the weights. The challenge with division is that even if numerator and denominator can be expressed in a small number of bits, their quotient may not be easily expressed (e.g., $\frac{11}{7}$).


We compute the total weight $t$ and assign each element a contiguous region in $[0,t)$ with length equivalent to its weight, (i.e., the first element is assigned $[0,w_1)$ the second $[w_1,w_1 +w_2)$ and so on).
The sampling method follows from the observation that sampling a value uniformly at random in $[0,t)$ is identical to sampling from the normalized distribution. To sample uniformly from $[0,t)$, we construct a string of random bits with maximum precision in $[0, 2^{g})$ where $2^g$ is the smallest power of $2$ greater than or equal to $t$. If the sample exceeds $t$, then we discard it and try again. We then compute the cumulative weights $c_i := \sum_{j \leq i}w_i$, and each element is assigned the region $[c_{i-1},c_i)$.
Given the sample $s \sim [0,t)$, we select the corresponding element index by 
performing a binary search on the range $[0,t)$ which discards elements ``ruled out'' by each bit of $s$, stopping when only one valid element remains. 

Why not simply iterate through the values and output the smallest $i$ such that $c_i> s$? If the precision provided is correct, there is no difference between these two approaches. However, if the precision provided is somehow incorrect, (e.g., because a bug was introduced in precision determination, the weights do not actually conform to the bounds used to compute the precision, the inexact arithmetic monitoring fails, etc.), the ``binary search'' style method identifies if the precision used for sampling was insufficient to correctly sample from the desired distribution, as more than one element will remain in this case.
Algorithm \ref{alg:psweightedsample} outlines this procedure in detail.

\begin{algorithm}[htbp]
  \caption{Normalized weighted sampling}
  \label{alg:psweightedsample}
\begin{algorithmic}[1]
  \LineComment{\textbf{Inputs}: $W$, a set of weights, $p$, the precision and $k$, a timing channel mitigation parameter.} 
  \LineComment{\textbf{Outputs}: $i$, an index sampled according to $p(i) := \frac{w_i}{\sum_{w_j \in W} w_j}$.}
    \Procedure{NormalizedSample}{$W,p,k$}
    \State $t \leftarrow \sum_{w \in W}w$ \Comment{The total weight.}

    \For{$i \in \{1,\ldots,|W|\}$} \Comment{Compute the cumulative weights.}
       \State $c_i \leftarrow \sum_{j=1}^i w_j$  \Comment{Each element assigned $[c_{i-1},c_i)$.}
    \EndFor
    \State $s \leftarrow 0$
    \State $j  \leftarrow 0$ 
    \State $g \leftarrow \argmin_g \{2^g \geq t\}$ 
    \State $R \leftarrow [|W|]$ \Comment{the remaining elements}
    \State $s^* \leftarrow \textsc{GetRandomValue}(p,t)$
    \While{$|R| > 1$} \label{line:loop}
        \State $s \leftarrow s + s^*_{g-j}2^{g-j}$ \Comment{Add the $j^{th}$ bit of $s^*$ to $s$.}
        \For{$i \in R$}
            \If{$c_i \leq s $} \Comment{$s$ cannot be in $[c_{i-1},c_i)$, even if all remaining bits of $s^*$ are $0$.}
                \State $R \leftarrow R \backslash \{i\}$ 
            \EndIf
            \If{$i>0$ and $c_{i-1} \geq s + 2^{g-j} $} \Comment{$s$ cannot be in $[c_{i-1},c_i)$, even if all remaining bits of $s^*$ are 1.}
                \State $R \leftarrow R \backslash \{i\}$ 
            \EndIf
        \EndFor
        \State $j \leftarrow j + 1$
        \If{$|R| > 1$ and $j > p$} \Comment{Precision was insufficient} \label{line:ps}
            \State \textbf{return} error
        \EndIf
    \EndWhile
    \State \textbf{return} $l$ 
    \EndProcedure
    \LineComment{\textbf{Inputs:} the number of bits of precision $p$, the upper bound of the range $[0,t)$ from which to sample, and $k$, a timing channel mitigation parameter.}
    \LineComment{\textbf{Output:} a number sampled uniformly at random from $[0,t)$.}
    \Function{GetRandomValue}{$p,t,k$}
    \State $g \leftarrow \argmin_g \{2^g \geq t\}$
    \State Initialize $s \leftarrow \infty$,  $\mathbf{s^* \leftarrow \infty}$, ${c \leftarrow 0}$
    \While{$s \geq t$ and ${c < k}$} \Comment{Reject $s$ if it falls outside $[0,t)$ or fewer than $k$ iterations.} \label{line:whilereject}
        \For{$i \in \{0,\ldots,p\}$}
            \State $r_i \sim \mathbf{Unif}(0,1)$ 
        \EndFor
        \State $s \leftarrow \sum_{i =0}^p r_i2^{g-i}$ \label{line:sum}
        \If{$s^* = \infty$ and $s < t$} 
            \State $s^* \leftarrow s$ \Comment{update $s^*$ if $s$ is in range.}
        \EndIf
        \State $c \leftarrow c + 1$
    \EndWhile
    \State $\textbf{return}$ $s^*$
    \EndFunction
  \end{algorithmic}
\end{algorithm}

\clearpage

\begin{lemma}\label{lemma:normalizesampling}
Given a set of weights $W$ with total weight $t$ and sufficient precision $p$ such that the addition of any combination of weights in $W$ can be expressed exactly with precision $p$, 
Algorithm \ref{alg:psweightedsample} 
outputs the index of an element in $W$ with probability distribution identical to that of the exponential mechanism and
the procedure uses at most $O(p)$ random bits with high probability. If $p$ is not sufficient to sample correctly from the distribution, an error is returned.
\end{lemma}
The proof of Lemma \ref{lemma:normalizesampling} appears in the appendix.

\subsubsection*{Timing Channels}
Although it is subtle, Algorithm \ref{alg:psweightedsample} has a timing channel, driven by the difference in rejection probability (Line \ref{line:whilereject}).  For example, if the total weight is $2^\ell$ for database $d$, but for $d'$ the total weight is $2^{\ell}-1$, the sampling procedure for $d$ may take longer than the sampling procedure for $d'$. This follows from observing that the probability of $s\geq t$ is zero for $d'$, but near $\frac{1}{2}$ for $d$. The failure probability always ranges from $[0,\frac{1}{2}]$, so it cannot be arbitrarily different even for very large differences in the total weight. The severity of this timing channel depends on how quickly random bits can be sampled and how easily an adversary can distinguish between different numbers of rejections. 
(We demonstrate a practical exploit of this timing channel in Section \ref{sec:implementation}.)

For example, suppose that for $d$, the rejection probability is $\frac{1}{2}$ and for $d'$ the rejection probability is $0$. With probability greater than $0.999$, $d$ will reject fewer than 10 times, so if generating the random value 10 times versus one time is indistinguishable by the adversary, then the timing channel will not be useful. However, if the adversary has access to fine-grained timing information (e.g., if the adversary is able to make randomness scarce and thus increase the time it takes to sample random bits) or the number of random bits used, this channel could be problematic.\footnote{Recall that $p$ is set independent of the database, so we are not concerned if the adversary learns the number of bits of precision required for each iteration of the loop, rather we are concerned about the adversary learning the number of times the loop is executed.} Furthermore, it's worth noting that this timing channel is not dependent on the privacy parameter, so the adversary may repeat the procedure many times with small budget in order to observe timing information. 

Fortunately, we can mitigate this timing channel by forcing the rejection loop to run at least $k$ times, even if a suitable $s$ has already been identified. By running the loop a fixed minimum number of times, the timing channel is not completely removed, but the probability of observing more than $k$ failures, and thus that any difference in timing can be observed, can be driven arbitrarily low by increasing the size of $k$. 

\begin{proposition}\label{prop:timingchannel}
Fixing $p$, for any $t$, the probability that \textsc{GetRandomValue} executes its \textbf{while} loop more than $k$ times is bounded by $2^{-k}$, thus with probability at least $1-2^{-k}$, \textsc{GetRandomValue} has the same running time regardless of the underlying database. Thus for any adversary that can observe fine-grained timing information or the number of random bits used, the probability that an adversary observes useful timing information from \textsc{GetRandomValue} is at most $2^{-k}$.  
\end{proposition}

\noindent It is critical to note that the timing channel mitigated is due to differences in logic, not differences in the underlying timing of floating point arithmetic operations, which could still be a source of timing discrepancy, (e.g., \cite{andrysco2015subnormal}).

A second potential source of timing channel is the concentration of weight, i.e., how many elements can be eliminated in each iteration of the loop in Line \ref{line:loop}. Fortunately, the implementation of  Algorithm \ref{alg:psweightedsample} can be modified to use all of the bits of $s^*$ in the first iteration, and make a single \textit{complete} pass through the cumulative weights. More concretely, this modification consists of replacing Line 11 with $s \leftarrow s^*$, and setting the initial value of $j \leftarrow p$ in Line 6. This modification ensures that the amount of work done by the implementation is the same, even if the concentration of the weights is different. This procedure can be optimized by terminating the scan of the cumulative weights early, at the cost of re-introducing the weight concentration timing channel.

\subsubsection*{Optimization}
Although timing channels are problematic in many cases, if the procedure is run in a batch mode or a setting in which no fine-grained timing or randomness use information can be accessed by an adversary, it is possible to optimize the sampling procedure at the cost of a more significant timing channel. Simply put, one need not draw all $p$ bits of randomness at once, and can instead draw one bit at a time and terminate once a single value has been identified. Algorithm \ref{alg:optimizedweightedsample} in the appendix outlines this optimization in detail.

\subsection{Putting it all together}
We now have all of the ingredients to specify the implementation of the exponential mechanism. Algorithm \ref{alg:b2exp} outlines the procedure. We leave the error handling logic in the case that inexact arithmetic occurs unspecified. Depending on timing channel concerns, the database curator could choose to retry at higher precision or to take some other action to debug what went wrong. As specified, inexact arithmetic or errors should only occur if there are bugs in the implementation.

\begin{algorithm}[htbp]
  \caption{Base-2 exponential mechanism}
  \label{alg:b2exp}
\begin{algorithmic}[1]
  \LineComment{\textbf{Inputs}: data independent parameters $u_{min}$, $u_{max}$, $\eta(x,y,z)$, $o_{max}$, data dependent parameters $u$ a utility function, $O$ the set of outcomes and $k$ a timing channel mitigation parameter.} 
  \LineComment{\textbf{Outputs}: $o \in O$ sampled according to the probability distribution of the base-2 exponential mechanism or an error.}
    \Procedure{B2ExponentialMechanism}{$u_{min}$, $u_{max}$, $o_{max}$, $\eta(x,y,z)$, $u$, $O$, $k$}
    \State Check that $\eta(x,y,z)$ is a valid parameter choice
    \State $p \leftarrow (\max(1,|u_{min}|)  + \max(1,|u_{max}|))z(y + b_x) + o_{max}$
    \State Initialize an empty utility list $U$
    \For{$o \in O$}
    \State Append $clamp(u(o),u_{min},u_{max})$ to $U$
    \EndFor
    \State Begin monitoring inexact arithmetic
    \State $b \leftarrow 2^{-\eta}$
    \State Initialize an empty weight list $W$
    \For{$i \in |O|$}
        \State $w_o \leftarrow b^{U[i]}$
        \State Append $w_o$ to $W$
    \EndFor
    \State $i^* \leftarrow \textsc{NormalizedSample}(W,p,k)$
    \If{inexact arithmetic or $i^{*}$ is error}
        \State (Error handling logic)
    \EndIf
    \State $o^* \leftarrow O[i^*]$
    \State \textbf{return} $o^*$
    \EndProcedure
  \end{algorithmic}
\end{algorithm}

We state the following proposition to formally characterize the behavior of Algorithm \ref{alg:b2exp}:

\begin{proposition}
Given parameters $u_{min}, u_{max}, o_{max}, \eta, O, k$ and $u$ determined independently of the database such that $\Delta u \leq 1$, 
\textsc{B2ExponentialMechanism} either (1) outputs an element from $O$ sampled according to the probability of the Exponential Mechanism on the utility function $clamp(u,$ $u_{min},u_{max})$ or (2) outputs an error if the precision is insufficient or inexact arithmetic occurs. Furthermore, except with probability $2^{-k}$, the mechanism does not provide a useful timing or randomness side channel if precision is sufficient and no inexact arithmetic occurs.
\end{proposition}

\section{Implementation}\label{sec:implementation}

Reference implementations in Python and Rust are included in the accompanying~\github. Please see the repository for detailed code structure and additional documentation.

\subsubsection*{Exact Arithmetic}
The goal of our implementation is to ensure that all arithmetic is exact. To that end, we use the GNU Multiple Precision Arithmetic (GMP) and the GNU MPFR libraries for high precision integer and floating point arithmetic \cite{Granlund12,fousse2007mpfr}. We access these libraries via the $\mathtt{gmpy2}$ interface for Python and the $\mathtt{gmp-mpfr-sys}$ and $\mathtt{Rug}$ crates for Rust \cite{gmpy2,SpiteriRug, SpiteriGMP}. 
MPFR provides a helpful set of flags indicating when operations are inexact, or result in overflow, underflow, truncation, etc. We use these extensively to ensure that our implementations are correct and exact. Under normal operation, we do not expect any exceptions to be raised due to inexact arithmetic. These libraries also provide methods for determining the maximum precision achievable on a given system at run-time, preventing any errors due to exceeding any system-specific thresholds and making the code easier to audit. In the case of an exception, the system designer can either choose to retry with higher precision, or take the failure as an indication that debugging is needed. 
We note that the problem of errors and exceptions being used as a side channel in DP is well known, and other techniques such as subsample and aggregate can help provide privacy-preserving error handling for these cases. The critical point is that the exception raised is the appropriate signal to the database administrator that an error has been exploited either intentionally or unintentionally.

\subsubsection*{Separation of data independent and dependent logic}
We have prioritized minimizing and clearly documenting side channels and our implementations endeavor to separate logic that should be data independent from data dependent logic. In particular, we require the caller to specify maximum and minimum utility bounds ($u_{min}$ and $u_{max}$) and a limit on the size of the outcome space ($o_{max}$), which allows us to  determine the working precision before computing any utilities or examining private data. This ensures that the working precision, which influences running time significantly, does not introduce a side channel. We also use these bounds to enforce limits on the size of the outcome space and to clamp outputs of the provided utility function to a prespecified range. Note, we do not attempt to verify the sensitivity of any provided utility function, and it is up to the caller to ensure appropriate sensitivity.

\subsubsection*{Specifying privacy parameters}
Privacy parameters are expressed as a set $\mathtt{\{x,y,z\}}$ to ensure that $2^{\eta}$ for $\eta = -z\log(\frac{x}{2^y})$ can be computed exactly.

\subsubsection*{Source of randomness}
Our implementations permit significant flexibility in the source of randomness used for sampling. In Python, the random number generator is passed in as a function pointer and is expected to produce individual random bits. In Rust, we use the $\mt{ThreadRandGen}$ trait defined by $\mt{Rug}$, which requires a $\mt{gen}$ function producing a random $\mt{u32}$. This allows the user to write simple wrappers for any desired randomness source. 


\subsubsection*{Logic}
The logic of our implementations is nearly identical to the naive exponential mechanism, but  includes randomized rounding of the utility function, utility clamping (i.e., clamping utilities to the pre-specified range) as well as additional monitoring of exact arithmetic.

\subsubsection*{Language-specific Details}
We provide reference implementations written in Rust and Python. (See Figures \ref{fig:py_usage} and \ref{fig:rust_usage}.) 
Although the logic of the implementations is for the most part identical, there are several differences worth noting. In Python we implement the exponential mechanism as a class with associated sampling and initialization functions. The intent of the implementation is to explicitly demonstrate the separation between data dependent and independent logic. For example, the caller first initializes an instance of the $\mathtt{ExpMech}$ class with the utility bounds, privacy parameter and maximum outcomes space size and then sets the utility function (which accesses private data) in a subsequent call, before sampling from the mechanism.   The Python mechanism also does not implement the timing channel or insufficient precision logic exactly specified in Algorithm \ref{alg:psweightedsample}, although it does include optimized sampling as specified in Algorithm \ref{alg:optimizedweightedsample}. The Python implementation is sufficiently close to the pseudocode given in Algorithm \ref{alg:psweightedsample} to be suitable for comparison with the Rust implementation which faithfully reproduces the algorithms as specified.

\begin{figure}[h]
    \centering
\begin{lstlisting}[language=Python]
# Example Python usage
from exponential_mechanism.expmech import *
import numpy as np

# Initialize outcomes and utility
u = lambda x : abs(0 - x)
# outcome space [-8, 8] in 0.3125 increments
gamma = 2**(-5)
k = 256
O = [i*gamma for i in range(-k,k+1)] 
# random bit generator using numpy randomness
rng = lambda : np.random.randint(0,2)

# Initialize the mechanism with eta = -log_2(1/2)
# and utility bounds
e = ExpMech(rng,eta_x=1,eta_y=1, eta_z=1,\ 
            u_min=0, u_max=16, max_o=len(O))
# set the utility function after initialization of precision, etc.
e.set_utility(u) 
# run the mechanism
result = e.exact_exp_mech(O)
\end{lstlisting}
    \caption{Example usage of Python mechanism to sample from the range $[-8,8]$ in increments of $2^{-5}$ for the utility function $u(o) = |0.0-o|$.}
    \label{fig:py_usage}
\end{figure}

\begin{figure}[H]
    \centering
\begin{lstlisting}[language=Rust]
// Example Rust usage
use b2dp::{exponential_mechanism, Eta, GeneratorOpenSSL};

// Define a utility function
fn util_fn (x: &u32) -> f64 {
     return (*x as f64).abs();
}
// Construct a privacy parameter eta = -log_2(1/2)
let eta = Eta::new(1,1,1).unwrap(); 
// Set bounds on the utility and outcomes
let utility_min = 0; 
let utility_max = 10;
let max_outcomes = 10;
let rng = GeneratorOpenSSL {}; // use OpenSSL for randomness
let outcomes: Vec<u32> = (0..max_outcomes).collect();
let sample = exponential_mechanism(eta, 
                        &outcomes, util_fn, utility_min, 
                        utility_max, max_outcomes, rng, 
                        Default::default())
                        .unwrap();
\end{lstlisting}
    \caption{Example usage of Rust mechanism with utility function $u(o) = |o|$ and $O = \{0,\ldots, 10\}$ and default options for optimization and timing channel protections.}
    \label{fig:rust_usage}
\end{figure}
The Rust implementation prioritizes performance in addition to readability and simplicity of the code, and is intended for longer-term development. It separates  functionality like normalized sampling and exact arithmetic monitoring from the implementation of the exponential mechanism itself in order to make extensions to other applications easier in the future. The Rust implementation includes the minimum retry and insufficient precision logic outlined in Algorithm \ref{alg:psweightedsample}.


As we discuss below, each language has different overheads, and the performance characteristics of the sampling procedures vary between the implementations. 
Furthermore, each language relies on different interfaces to the underlying high precision arithmetic libraries, and use slightly different strategies for enforcing exact arithmetic. 
In Python, we configure the high precision arithmetic libraries to raise an exception if inexact arithmetic is performed. In Rust, we instead access the inexact arithmetic flags directly, and enforce exact arithmetic only at the end of the procedure. 

\subsection{Performance and Applications}\label{section:laplace}

All timing tests were run using the $\mathtt{Criterion}$ crate for Rust and a custom timing script for Python on a Linux virtual machine with 2 cores and 2GB of RAM. We used randomness generated by $\mathtt{OpenSSL}$ in the Rust configurations, and $\mathtt{numpy}$ randomness in Python. In general, Rust far outperformed Python, but the times presented are not strictly comparable, given the differences in benchmarking. However, the times presented are indicative of performance on low-end consumer hardware. We tested 
several configurations, outlined in Table \ref{table:configs}. Unless otherwise noted, we chose the retry parameter $k=1$ for timing channel mitigation.
\begin{table}[t]
  \centering%
  \begin{tabular}{|l|p{13cm}|}\hline
    $\mathtt{PythonNaive}$ & naive base-$e$ Python implementation \\ \hline
    $\mathtt{Python}$ & base-2 Python. \\ \hline
    $\mathtt{PythonOpt}$ & base-2 Python using Algorithm \ref{alg:optimizedweightedsample} for weighted sampling. \\ \hline
    $\mathtt{PythonEmp}$ & $\mathtt{Python}$ using empirical precision (Algorithm \ref{alg:precision}).\\ \hline
    $\mathtt{RustNaive}$ & naive base-$e$ Rust implementation \\ \hline
    $\mathtt{Rust}$ & base-2 Rust. \\ \hline
    $\mathtt{RustOpt}$ & base-2 Rust using Algorithm \ref{alg:psweightedsample} with early termination for weighted sampling. \\ \hline
    $\mathtt{RustEmp}$ & $\mathtt{Rust}$ using empirical precision (Algorithm \ref{alg:precision}). \\ \hline
  \end{tabular}
  \caption{Testing configurations}\label{table:configs}
\end{table}

Interestingly, optimized sampling (Algorithm \ref{alg:optimizedweightedsample}) showed degraded performance compared with Algorithm \ref{alg:psweightedsample} in Rust. This is likely due to the increased overhead of the binary search loop in comparison with sampling additional bits of randomness in Rust. As such, we sample all randomness at once, and we use the early termination optimization for Algorithm \ref{alg:psweightedsample} in our optimized Rust configuration. 

We consider the following application test cases to illustrate different properties of the base-2 exponential mechanism implementations.

\begin{figure}
\centering
\begin{subfigure}{.45\textwidth}
  \centering
    \includegraphics[width=\textwidth]{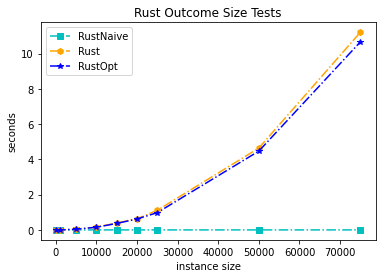}
    \caption{Rust: $O = [n]$ with $u(o)=o$.}
    \label{fig:rust_outcome}
\end{subfigure}%
\begin{subfigure}{.45\textwidth}
  \centering
    \includegraphics[width=\textwidth]{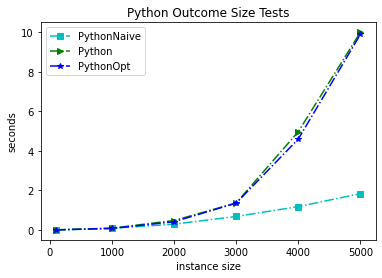}
    \caption{Python: $O = [n]$ with $u(o)=o$.}
    \label{fig:py_outcome}
\end{subfigure}
\caption{}
\label{fig:outcomes}
\end{figure}


\subsubsection*{Outcome Space Sizes}
The simplest test we consider is choosing an outcome set $O = [n]$ with a utility function $u(o)=o$. As the instance size $n$ increases, both the utility range and the outcome space size increase. Figures \ref{fig:rust_outcome} and \ref{fig:py_outcome} illustrate the results for Python and Rust. Although $\mathtt{Python}$ and $\mt{PythonOpt}$ are competitive with $\mt{PythonNaive}$, $\mt{RustNaive}$ is orders of magnitude faster than $\mt{Rust}$ and $\mt{RustOpt}$. (We omit $\mt{RustNaive}$ and $\mt{PythonNaive}$ from later figures with the exception of the Laplace test case, as the conclusions are not significantly different in other test cases.) This test demonstrates that the mechanism (particularly in Rust) is practical for large instance sizes, as even in a low-end configuration, instance sizes of $75$k can be completed in about 10 seconds.

\begin{figure}
\centering
\begin{subfigure}{.45\textwidth}
    \centering
    \includegraphics[width=\textwidth]{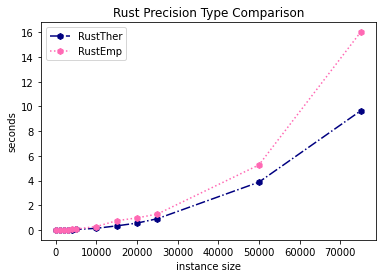}
    \caption{Rust: $O = [n]$ with $u(o)=o$.}
    \label{fig:rust_precision}
\end{subfigure}%
\begin{subfigure}{.45\textwidth}
    \centering
    \includegraphics[width=\textwidth]{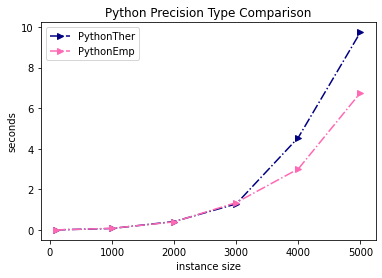}
    \caption{Python: $O = [n]$ with $u(o)=o$.}
    \label{fig:py_precision}
\end{subfigure}
\caption{}
\label{fig:precision}
\end{figure}


\subsubsection*{Precision Determination Method}
We repeated the outcome space size tests for $\mt{RustEmp}$ and $\mt{PythonEmp}$, to demonstrate the difference in performance due to the relative trade-off between the computational overhead of the empirical precision determination versus the reduction in precision. Figures \ref{fig:rust_precision} and  \ref{fig:py_precision} outline the results. 
$\mathtt{PythonEmp}$ significantly outperforms $\mathtt{PythonTher}$ as the instance size increases, indicating that the computational overhead of the precision determination is outweighed by the improved sampling performance. In contrast, $\mathtt{RustEmp}$ performs significantly worse, indicating that the precision determination overhead is not outweighed by the reduced precision. (The Laplace use case, below, has different behavior depending on the precision determination method, but no noticeable difference is observed for the Utility Ranges test.)

\subsubsection*{Utility Ranges}
Given a fixed outcome space size $|O|=k+1=1000$, we test how increasing the range of utilities $[0,1000+n]$ impacts performance. We take $O = \{0\} \cup \{i + n \mid  i \in [k]\}$ and assign $u(o)=o$.
Figures \ref{fig:rust_range} and \ref{fig:py_range} illustrate the outcomes of the test for the optimized and non-optimized Rust and Python implementations. As expected, the growth curve is linear, as the only increase in cost as the instance size $(n)$ increases is the cost of using increased precision for computation. This test shows the benefit of optimized sampling in Python, and the performance of $\mathtt{PythonOpt}$ diverges from the performance of $\mathtt{Python}$ as the instance size increases. 

\begin{figure}[h]
\centering
\begin{subfigure}{.45\textwidth}
     \centering
    \includegraphics[width=\textwidth]{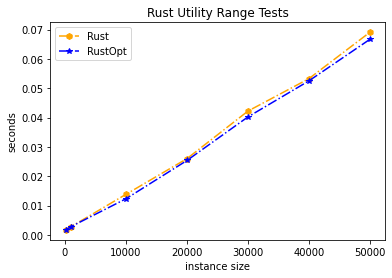}
    \caption{Rust: $k=1000$, $O = \{0\} \cup \{i + n $ $| i \in [k]\}$, $u(o)=o$.}
    \label{fig:rust_range}
\end{subfigure}%
\quad 
\begin{subfigure}{.45\textwidth}
    \includegraphics[width=\textwidth]{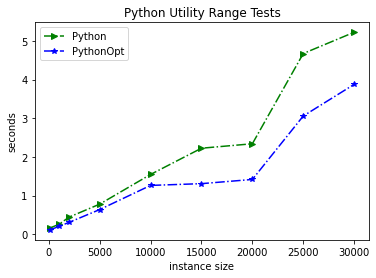}
    \caption{Python: $k=1000$, $O = \{0\} \cup \{i + n | i \in [k]\}$, $u(o)=o$.}
    \label{fig:py_range}
\end{subfigure}
\caption{}
\label{fig:range}
\end{figure}

\subsubsection*{Laplace}
Recall that the Laplace mechanism (Mechanism \ref{mech:laplace}) computes a function $f$ of the database and then adds noise drawn from the Laplace distribution, described by probability density function $\mathsf{Lap(t|\frac{\Delta f}{\varepsilon})}=\frac{\varepsilon}{2\Delta f}e^{-\frac{\varepsilon |t|}{\Delta f }}$. 
Notice that the range of possible noise values is infinite - although the probability of observing a sample $\geq 14$ is less than one in a million for $\varepsilon=1$. 
The standard method for sampling from the Laplace distribution is to use inverse transform sampling, i.e., to sample a uniform value in $U \in [0,1]$ and solve for $CDF(t)=U$, where $CDF$ refers to the cumulative density function of the Laplace distribution.\footnote{Recall that the cumulative density function is the integral of the probability density function from $-\infty$ to $t$, and in the case of Laplace is very simple to compute.} In a simple implementation, this amounts to computing $\tau = \frac{\varepsilon}{2\Delta f}\ln(1-U)$ for a uniformly random values $U \in [0,1)$.\footnote{See \cite{mironov2012significance} for a longer discussion of this technique for the Laplace mechanism in the context of DP and \cite{ridout2009generating} for a more general discussion of sampling from the Laplace distribution.}

As discussed extensively in \cite{mironov2012significance}, the danger of implementing this mechanism using standard (even highly accurate) implementations of $\ln$ is that there are certain \textit{gaps} in the low order bits of the values produced by $\ln(x)$ for $x \in [0,1]$ even when $x$ is drawn at regular intervals. These gaps result in differences in the support (possible outcomes) of the mechanism on adjacent databases, which can then be used to distinguish possible values of the original output of $f$. By ``snapping'' $f + \tau$ to a fixed set of intervals, these low order bits are safely eradicated. (Further details on implementing the snapping mechanism can be found in \cite{Covington19}.)

\begin{figure}[t]
\centering
\begin{subfigure}{.45\textwidth}
    \centering
    \includegraphics[width=\textwidth]{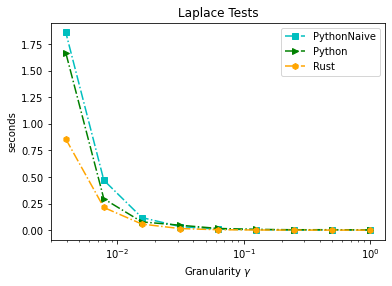}
    \caption{Comparison of $\mt{PyNaive}$, $\mt{Python}$ and $\mt{Rust}$ for $O=\{B_L + \gamma_i \mid B_L + \gamma_i \leq B_U\}$, $u(o)=|0.0-o|$.}
    \label{fig:combined_laplace}
\end{subfigure}%
\quad
\begin{subfigure}{.45\textwidth}
    \centering
    \includegraphics[width=\textwidth]{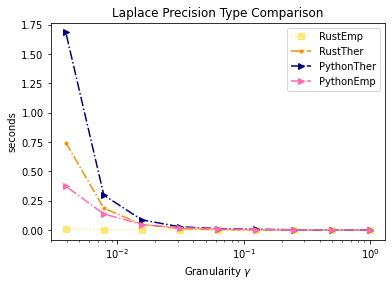}
    \caption{Comparison of $\mt{PythonEmp}$, $\mt{Python}$, $\mt{RustEmp}$ and $\mt{Rust}$ for $O=\{B_L + \gamma_i \mid B_L + \gamma_i \leq B_U\}$, $u(o)=|0.0-o|$.}
    \label{fig:combined_laplace_precision}
\end{subfigure}
\caption{}
\label{fig:laplace}
\end{figure}

However, the Laplace mechanism can also be simulated via the exponential mechanism by setting the utility function to be $u(d,o) = |f(d) - o|$, although there is some subtlety in the selection of $O$. We implement the clamped discrete Laplace mechanism, which takes in a lower bound ($B_L$), an upper bound ($B_U$) and a granularity ($\gamma$) and sets $O = \{B_L + \gamma_i \mid B_L + \gamma_i \leq B_U\}$. By setting this outcome space before the target value $f(d)$ is observed, we maintain the independence of the outcome set and the utility function. For our tests, we chose $B = 10$, and varied the size of $\gamma$. Figure \ref{fig:combined_laplace} compares the results of comparing $\mt{Rust}$ with and $\mt{Python}$ with $\mt{PythonNaive}$. In this case, $\mt{Python}$ actually outperforms $\mt{PythonNaive}$, likely due to the improved performance of high precision arithmetic in MPFR/$\mt{gmpy2}$ versus built-in types. Furthermore, as illustrated in Figure \ref{fig:combined_laplace_precision},  $\mt{PythonEmp}$ outperforms $\mt{Python}$ \textit{and} $\mt{Rust}$. Interestingly, $\mt{RustEmp}$ performs extremely well. This is likely due to the limited outcome space sizes (which reduces the overhead of the empirical precision computation). Given the differences in the benchmarking setup, we don't read too much into the exact timing comparison, but it is interesting to note that this is the only testing scenario in which Python approaches the performance of Rust. This test has a very limited utility range and outcome space size ($<7$k), and as such the benefits of the Rust implementation seem to be outweighed by the basic overheads of running the mechanism. 

Other than performance, a critical concern with using the base-2 exponential mechanism for Laplace is how much the use of randomized rounding impacts the accuracy of the mechanism. In brief, randomized rounding has a limited impact and the rounded mechanism behaves very similarly to the unmodified, inexact mechanism in the case of simulating the Laplace distribution. Furthermore, a small outcome set size ($<1,000$) is sufficient to closely approximate the true Laplace distribution.

Figure \ref{fig:rrcomparisons} illustrates the distribution of outputs simulating the Laplace distribution for the randomized rounding mechanism and the naive mechanism. 
The distributions are nearly identical, and  Kolmogorov-Smirnov test statistic for the pair of distributions is also small $(0.02)$, indicating that there isn't sufficient evidence to conclude that the distributions are different. In the Appendix, 
Figures \ref{fig:pointwise} and \ref{fig:boundcomparison} illustrate how closely the lower and upper bounds given by Lemma \ref{lemma:heuristic} bound the target probability of the naive mechanism. In particular, the width of the bound is closely related to the granularity of the outcome space, and as granularity increases, the width of the bound decreases, giving smaller point-wise error bounds.

\begin{figure}%
\centering
\label{fig:rrvsnaive}%
\includegraphics[width=0.45\textwidth]{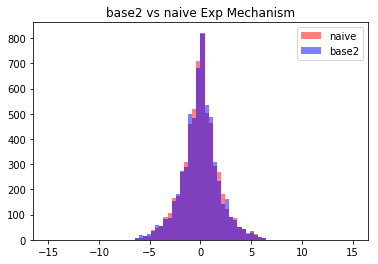}
\caption{A comparison of the outputs of the randomized rounding base-2 exponential mechanism with the unmodified exponential mechanism (the naive mechanism) for sampling from the Laplace distribution for $n=6000$ trials. 
($\varepsilon=2\ln(2)$, $\gamma=2^{-4}$, $O = \{\gamma i, -\gamma i \mid i \in [100]\}$.)}
\label{fig:rrcomparisons}
\end{figure}


\subsubsection*{Timing channel protections}
To demonstrate the impact of minimum retry logic on timing channels, we constructed two utility functions $u_1$ such that $u_1(o) = 1$ and $u_0$ such that $u_0(o_1) = 0$ and $u_0(o_{i>1}) = 1$ and ran them on an outcome space of size $256$ with $\eta = -\log_2(\frac{1}{2})$. This results in a total weight of $128.5$ for $u_0$ and a total weight of $128$ for $u_1$. Thus, $u_1$ will always succeed in choosing a sample in $[0,128)$, but $u_0$ will have rejection probability near $\frac{1}{2}$.
Figure \ref{fig:timing_channel} illustrates the difference in timing for these two utility functions with varying $k$ dictating the minimum number of retries. As we would expect from Proposition \ref{prop:timingchannel}, although there is a noticeable timing difference when $k=1$, as $k$ increases, the timing converges. This both demonstrates the potential for a malicious adversary to exploit the timing channel (if protections are not in place), and that the minimum rejection method does provide empirically validated protection.
\begin{figure}
    \centering
    \includegraphics[width=0.5\textwidth]{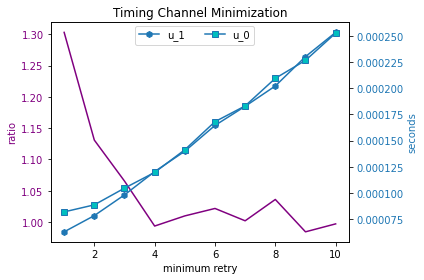}
    \caption{Rust: Comparison of timing differences of utility functions $u_0(o)=0$ for all $o \in O$ and $u_1(o_{i>1} = 0)$, $u_1(o_1)=1$ with privacy parameter set such that the total weight for $u_1=128.5$ and for $u_0=128$. $u_0$ always successfully samples within the required range, whereas $u_1$ has near maximal rejection probability. With no retries, $u_1$ is consistently about $1.5\times$ slower, but as the number of minimum retries increases, the times converge.}
    \label{fig:timing_channel}
\end{figure}

\section{Discussion and future work}
Floating-point issues in implementations of the Laplace mechanism are well-known. In this work, we have demonstrated new floating point issues which arise in the exponential mechanism even when the floating point computations themselves are not revealed. To address these issues, we have specified an alternative view of DP using base-2 instead of base $e$ exponentiation. 
Base-2 DP has the benefit of allowing exact implementation of the exponential mechanism, without diverging from the logic of the original base-$e$ mechanism. 
By using exact computations, monitoring of correctness, which yields privacy, is more straightforward, and the mechanism is much harder to subvert as it is not vulnerable to attacks exploiting inexact arithmetic. We have also shown that exact implementation permits a rich set of privacy parameters and both integer and non-integer utility functions. The base-2 mechanism is reasonably efficient. Even on low-end consumer hardware, our reference implementation can handle instance sizes on the order of $100$k in seconds.



This work is a step towards a programming paradigm of exact arithmetic for DP. Ongoing developments (available in the accompanying \github) include: implementation of the integer partition mechanism (and extension to pure-DP) of \cite{blocki2016differentially} and extensions of floating point attacks to the sparse vector technique and mitigation via exact implementation of clamped discrete Laplace and noisy threshold.
Future directions include: integration with open-source DP libraries, e.g. OpenDP \cite{OpenDP}, to more easily construct bounded-sensitivity utility functions, more robust timing channel mitigation strategies, additional methods for converting between and composing base-2 and base-$e$ privacy parameters.

\bibliography{biblio}

\begin{thebibliography}{20}
\providecommand{\natexlab}[1]{#1}
\providecommand{\url}[1]{\texttt{#1}}
\expandafter\ifx\csname urlstyle\endcsname\relax
  \providecommand{\doi}[1]{doi: #1}\else
  \providecommand{\doi}{doi: \begingroup \urlstyle{rm}\Url}\fi

\bibitem[Adams(2013)]{adams2013Gumbel}
R.~Adams.
\newblock The gumbel-max trick for discrete distributions, 2013.
\newblock URL
  \url{https://lips.cs.princeton.edu/the-gumbel-max-trick-for-discrete-distributions/}.

\bibitem[Andrysco et~al.(2015)Andrysco, Kohlbrenner, Mowery, Jhala, Lerner, and
  Shacham]{andrysco2015subnormal}
M.~Andrysco, D.~Kohlbrenner, K.~Mowery, R.~Jhala, S.~Lerner, and H.~Shacham.
\newblock On subnormal floating point and abnormal timing.
\newblock In \emph{2015 IEEE Symposium on Security and Privacy}, pages
  623--639. IEEE, 2015.

\bibitem[Balcer and Vadhan(2017)]{balcer2017differential}
V.~Balcer and S.~Vadhan.
\newblock Differential privacy on finite computers.
\newblock \emph{arXiv preprint arXiv:1709.05396}, 2017.

\bibitem[Blocki et~al.(2016)Blocki, Datta, and
  Bonneau]{blocki2016differentially}
J.~Blocki, A.~Datta, and J.~Bonneau.
\newblock Differentially private password frequency lists.
\newblock In \emph{NDSS}, volume~16, page 153, 2016.

\bibitem[Burke(2017)]{burke2017}
M.~Burke.
\newblock Vulnerability in floating point implementation of exponential
  mechanism, 2017.
\newblock Poster, Theory and Practice of Differential Privacy Workshop.

\bibitem[Covington(2019)]{Covington19}
C.~Covington.
\newblock \emph{Snapping Mechanism}, 2019.
\newblock URL \url{https://github.com/ctcovington/floating_point}.

\bibitem[Dwork et~al.(2006)Dwork, McSherry, Nissim, and
  Smith]{dwork2006calibrating}
C.~Dwork, F.~McSherry, K.~Nissim, and A.~Smith.
\newblock Calibrating noise to sensitivity in private data analysis.
\newblock In \emph{Theory of cryptography conference}, pages 265--284.
  Springer, 2006.

\bibitem[Dwork et~al.(2014)Dwork, Roth, et~al.]{dwork2014algorithmic}
C.~Dwork, A.~Roth, et~al.
\newblock The algorithmic foundations of differential privacy.
\newblock \emph{Foundations and Trends{\textregistered} in Theoretical Computer
  Science}, 9\penalty0 (3--4):\penalty0 211--407, 2014.

\bibitem[Fousse et~al.(2007)Fousse, Hanrot, Lef{\`e}vre, P{\'e}lissier, and
  Zimmermann]{fousse2007mpfr}
L.~Fousse, G.~Hanrot, V.~Lef{\`e}vre, P.~P{\'e}lissier, and P.~Zimmermann.
\newblock Mpfr: A multiple-precision binary floating-point library with correct
  rounding.
\newblock \emph{ACM Transactions on Mathematical Software (TOMS)}, 33\penalty0
  (2):\penalty0 13--es, 2007.

\bibitem[Gazeau et~al.(2016)Gazeau, Miller, and
  Palamidessi]{gazeau2016preserving}
I.~Gazeau, D.~Miller, and C.~Palamidessi.
\newblock Preserving differential privacy under finite-precision semantics.
\newblock \emph{Theoretical Computer Science}, 655:\penalty0 92--108, 2016.

\bibitem[Granlund and {the GMP development team}(2019)]{Granlund12}
T.~Granlund and {the GMP development team}.
\newblock \emph{{GNU MP}: {T}he {GNU} {M}ultiple {P}recision {A}rithmetic
  {L}ibrary}, 6.1.2 edition, 2019.
\newblock URL \url{http://gmplib.org/}.

\bibitem[Gupta et~al.(2009)Gupta, Ligett, McSherry, Roth, and
  Talwar]{talwar2009approximation}
A.~Gupta, K.~Ligett, F.~McSherry, A.~Roth, and K.~Talwar.
\newblock Differentially private approximation algorithms.
\newblock \emph{CoRR}, abs/0903.4510, 2009.
\newblock URL \url{http://arxiv.org/abs/0903.4510}.

\bibitem[Ilvento(2020)]{ilvento2020}
C.~Ilvento.
\newblock \emph{Working Paper}. {Implementing} the sparse vector technique and
  variants: vulnerabilities and mitigation via exact implementation, 2020.

\bibitem[Martelli and {the gmpy2 contributors}(2019)]{gmpy2}
A.~Martelli and {the gmpy2 contributors}.
\newblock \emph{General Multi-Precision arithmetic for Python}, 2.1 edition,
  2019.
\newblock URL \url{https://gmpy2.readthedocs.io/en/latest/}.

\bibitem[McSherry and Talwar(2007)]{mcsherry2007mechanism}
F.~McSherry and K.~Talwar.
\newblock Mechanism design via differential privacy.
\newblock In \emph{FOCS}, volume~7, pages 94--103, 2007.

\bibitem[Mironov(2012)]{mironov2012significance}
I.~Mironov.
\newblock On significance of the least significant bits for differential
  privacy.
\newblock In \emph{Proceedings of the 2012 ACM conference on Computer and
  communications security}, pages 650--661. ACM, 2012.

\bibitem[Ridout(2009)]{ridout2009generating}
M.~S. Ridout.
\newblock Generating random numbers from a distribution specified by its
  laplace transform.
\newblock \emph{Statistics and Computing}, 19\penalty0 (4):\penalty0 439, 2009.

\bibitem[Spiteri and {gmp-mpfr-sys contributors}(2020)]{SpiteriGMP}
T.~Spiteri and {gmp-mpfr-sys contributors}.
\newblock \emph{gmp-mpfr-sys: Rust low-level bindings for GMP, MPFR and MPC},
  1.2.2 edition, 2020.
\newblock URL \url{https://docs.rs/gmp-mpfr-sys/1.2.2/gmp_mpfr_sys/}.

\bibitem[Spiteri and {Rug contributors}(2020)]{SpiteriRug}
T.~Spiteri and {Rug contributors}.
\newblock \emph{Rug: Arbitrary-precision numbers}, 1.8.0 edition, 2020.
\newblock URL \url{https://docs.rs/rug/1.8.0/rug/}.

\bibitem[{The OpenDP Team}(2020)]{OpenDP}
{The OpenDP Team}.
\newblock The {OpenDP} whitepaper, 2020.
\newblock URL
  \url{https://projects.iq.harvard.edu/files/opendp/files/opendp_white_paper_11may2020.pdf}.

\end{thebibliography}

\appendix
\section{Appendix}
\subsection{Additional Proofs}\label{appendix:proofs}

\subsubsection*{Minimum precision}
Proof of Lemma \ref{lemma:minprec}.

\begin{proof}
Writing $2^{-\eta u}$ requires no more than $\max(1,|u|)z(y + b_x)$ bits of precision. Thus to write the largest weight, corresponding to $u_{min}$, we require no more than $\max(1,|u_{min}|)z(y + b_x)$ bits. The largest possible value any combination of weights could take on is $o_{max} 2^{-\eta u_{min}}$, and adding $2^{-\eta u_{min}}$ to itself $o_{max}$ times requires an extra $o_{max}$ bits of precision, as each addition increases the precision by at most one bit. Thus, the largest possible combination requires $o_{max} + \max(1,|u_{min}|)z(y + b_x)$ bits. This corresponds to the largest possible number of bits needed before the decimal. To compute the largest possible number of bits needed after the decimal, we consider the smallest possible weight, $2^{-\eta u_{max}}$, which will require $|u_{max}|z(y + b_x)$ bits. Thus in the worst case, we require maximum precision on both sides of the decimal, yielding the desired maximum bound of $(\max(1,|u_{min}|) + \max(1,|u_{max}|))z(y + b_x) + o_{max}$.
\end{proof}

\paragraph{Privacy of Randomized Rounding.}
Proof of Lemma \ref{lemma:randomizedrounding}

\begin{proof}

Suppose to implement randomized rounding that we draw a number $s$ uniformly at random from $[0,1)$, and round up if $s \leq |u(d,o)-\lfloor{u(d,o)} \rfloor|$ (and otherwise round down). 
Fix a particular choice of $s$. Consider a pair of adjacent databases $d$ and $d'$ such that $u(d',o)>u(d,o)$. Notice that it is impossible for the rounding procedure to result in a difference in composed utility of more than $\alpha$. This follows from observing that there are two cases: either $\lfloor{u(d',o)} \rfloor - \lfloor{u(d,o)} \rfloor < \alpha$, in which case any rounding results in difference at most $\alpha$, or $\lfloor{u(d',o)} \rfloor - \lfloor{u(d,o)} \rfloor = \alpha$.

If $\lfloor{u(d',o)} \rfloor - \lfloor{u(d,o)} \rfloor = \alpha$, then $\lfloor{u(d',o)} \rfloor - \lfloor{u(d,o)} \rfloor \geq u(d',o) - u(d,o)$, and thus 
 $u(d',o)-\lfloor{u(d',o)} \rfloor \leq u(d,o)-\lfloor{u(d,o)} \rfloor$. 
Therefore $s$ is  in one of three regions:
\begin{enumerate}[leftmargin=*]
    \item $s \in [0, u(d',o)-\lfloor{u(d',o)} \rfloor]$, which results in both values rounded up,
    \item $s \in (u(d',o)-\lfloor{u(d',o)} \rfloor,u(d,o)-\lfloor{u(d,o)} \rfloor]$, which results in $u(d',o)$ rounded down and $u(d,o)$ rounded up, or
    \item $s \in (u(d,o)-\lfloor{u(d,o)} \rfloor,1)$, which results in both rounded down.
\end{enumerate}
Thus, for any $s$ rounding never results in a  difference between $u(d,o)$ and $u(d',o)$ greater than $\alpha$. (The symmetric argument follows for any $o$ such that $u(d,o)>u(d',o)$.) 
Take the utility function $u_S:=\rho(u(d,o))$ to be the utility function with fixed randomness $S$, i.e., the set of $s$ used for each rounding decision.
From the above, $\Delta u_S \leq \alpha$. Thus, the exponential mechanism with utility function $u_S$ is $2\alpha\varepsilon-$DP. Write $p_S(o)$ for the probability that the exponential mechanism with utility function $u_S$ outputs the element $o$. Taking $p(o)$ to be the probability that the randomized rounding exponential mechanism outputs $o$, we can therefore write
\[p(o) = \sum_{S \sim [0,1)^{|O|}}\Pr[S]p_S(o)\]
and for any adjacent database, we can write
\[p'(o) = \sum_{S \sim [0,1)^{|O|}}\Pr[S]p'_S(o)\]
where $S \sim [0,1)^{|O|}$ indicates the set of all possible random values $s \in [0,1)$ used for sampling. 
Because $\Delta u_S \leq \alpha$, we have that  $\frac{p_S(o)}{p'_S(o)}\leq e^{-2\alpha\varepsilon}$, so
\begin{align*}
  \frac{p(o)}{p'(o)}  &= \frac{\sum_{S \sim [0,1)^{|O|}}\Pr[S]p_S(o)}{ \sum_{S \sim [0,1)^{|O|}}\Pr[S]p'_S(o)} \\ &\leq \frac{\sum_{S \sim [0,1)^{|O|}}\Pr[S]e^{2\alpha\varepsilon}p'_S(o)}{ \sum_{S \sim [0,1)^{|O|}}\Pr[S]p'_S(o)}\\
  & = e^{2\alpha\varepsilon}
\end{align*}
\end{proof}

\paragraph{Monotonicity of Randomized Rounding.}
\begin{lemma}\label{lemma:monotone}
The randomized rounding exponential mechanism is monotone. That is, for $o_1,o_2 \in O$ if  $\Pr[\mathcal{M}(d) = o_1] \geq \Pr[\mathcal{M}(d) = o_2]$ then $\Pr[\mathcal{M}'(d) = o_1]\geq \Pr[\mathcal{M}(d) = o_2]$, where $\mathcal{M}$ is the unmodified exponential mechanism, and $\mathcal{M}'$ is the randomized rounding exponential mechanism (Mechanism \ref{mech:rrexpmech}).
\end{lemma}
\begin{proof}
Fix any set of rounding decisions for $o_{i\notin \{1,2\}}$, and take $S = \sum_{i \notin \{1,2\}}e^{-\varepsilon u(d,o)}$.

Notice that if $|u(d,o_1) - u(d,o_2)| \geq 1$, then rounding the utilities cannot change the relative ordering of the utilities. Thus, we only consider cases where $|u(d,o_2) - u(d,o_1)| \leq 1$.

For conciseness, write $p_i^\bot:=  |u(d,o_i) - \lceil{u(d,o_i)}\rceil |$, i.e., the probability that the utility is rounded down,  and $p_i^\top:=  1- p_i^\bot$, i.e., the probability that the utility is rounded up.
We can write the probability that $\M'$ outputs $o_1$ or $o_2$ as 
\begin{align*}
    \Pr[\M'(d) = o_1] =& p_1^\bot * p_2^\bot \frac{e^{-\varepsilon \lfloor{u(d(o_1)}\rfloor}}{S + e^{-\varepsilon \lfloor{u(d(o_1)}\rfloor} + e^{-\varepsilon \lfloor{u(d(o_2)}\rfloor}}  \\&+ p_1^\bot * p_2^\top \frac{e^{-\varepsilon \lfloor{u(d(o_1)}\rfloor}}{S + e^{-\varepsilon \lfloor{u(d(o_1)}\rfloor} + e^{-\varepsilon \lceil{u(d(o_2)}\rceil}} 
    \\&+p_1^\top * p_2^\bot \frac{e^{-\varepsilon \lceil{u(d(o_1)}\rceil}}{S + e^{-\varepsilon \lceil{u(d(o_1)}\rceil} + e^{-\varepsilon \lfloor{u(d(o_2)}\rfloor}} 
    \\& + p_1^\top * p_2^\top \frac{e^{-\varepsilon \lceil{u(d(o_1)}\rceil}}{S + e^{-\varepsilon \lceil{u(d(o_1)}\rceil} + e^{-\varepsilon \lceil{u(d(o_2)}\rceil}}\\
    \Pr[\M'(d) = o_2] =& p_1^\bot * p_2^\bot \frac{e^{-\varepsilon \lfloor{u(d(o_2)}\rfloor}}{S + e^{-\varepsilon \lfloor{u(d(o_1)}\rfloor} + e^{-\varepsilon \lfloor{u(d(o_2)}\rfloor}}  \\&+ p_1^\bot * p_2^\top \frac{e^{-\varepsilon \lceil{u(d(o_2)}\rceil}}{S + e^{-\varepsilon \lfloor{u(d(o_1)}\rfloor} + e^{-\varepsilon \lceil{u(d(o_2)}\rceil}} 
    \\&+p_1^\top * p_2^\bot \frac{e^{-\varepsilon \lfloor{u(d(o_2)}\rfloor}}{S + e^{-\varepsilon \lceil{u(d(o_1)}\rceil} + e^{-\varepsilon \lfloor{u(d(o_2)}\rfloor}} 
    \\& + p_1^\top * p_2^\top \frac{e^{-\varepsilon \lceil{u(d(o_2)}\rceil}}{S + e^{-\varepsilon \lceil{u(d(o_1)}\rceil} + e^{-\varepsilon \lceil{u(d(o_2)}\rceil}}
\end{align*}

To show monotonicity, it suffices to show that if $e^{-\varepsilon u(d,o_1)} \geq e^{-\varepsilon u(d,o_2)} $, then $\Pr[\M'(d) = o_1] - \Pr[\M'(d) = o_2] \geq 0$.

\begin{align*}
    \Pr[\M'(d) = o_1] &- \Pr[\M'(d) = o_2]  \\
    =&p_1^\bot * p_2^\bot \frac{e^{-\varepsilon \lfloor{u(d(o_1)}\rfloor} - e^{-\varepsilon \lfloor{u(d(o_2)}\rfloor}}{S + e^{-\varepsilon \lfloor{u(d(o_1)}\rfloor} + e^{-\varepsilon \lfloor{u(d(o_2)}\rfloor}}  \\&+ p_1^\bot * p_2^\top \frac{e^{-\varepsilon \lfloor{u(d(o_1)}\rfloor}- e^{-\varepsilon \lceil{u(d(o_2)}\rceil}}{S + e^{-\varepsilon \lfloor{u(d(o_1)}\rfloor} + e^{-\varepsilon \lceil{u(d(o_2)}\rceil}} 
    \\&+p_1^\top * p_2^\bot \frac{e^{-\varepsilon \lceil{u(d(o_1)}\rceil}- e^{-\varepsilon \lfloor{u(d(o_2)}\rfloor}}{S + e^{-\varepsilon \lceil{u(d(o_1)}\rceil} + e^{-\varepsilon \lfloor{u(d(o_2)}\rfloor}} 
    \\ &+ p_1^\top * p_2^\top \frac{e^{-\varepsilon \lceil{u(d(o_1)}\rceil} - e^{-\varepsilon \lceil{u(d(o_2)}\rceil}}{S + e^{-\varepsilon \lceil{u(d(o_1)}\rceil} + e^{-\varepsilon \lceil{u(d(o_2)}\rceil}}
\end{align*}
As $e^{-\varepsilon \lceil{u(d(o_2)}\rceil} \leq e^{-\varepsilon \lceil{u(d(o_1)}\rceil}$ and $e^{-\varepsilon \lfloor{u(d(o_2)}\rfloor} \leq e^{-\varepsilon \lfloor{u(d(o_1)}\rfloor} $,
\begin{align*}
    \Pr[\M'(d) = o_1] &- \Pr[\M'(d) = o_2]  \\
    &\geq p_1^\bot * p_2^\top \frac{e^{-\varepsilon \lfloor{u(d(o_1)}\rfloor}- e^{-\varepsilon \lceil{u(d(o_2)}\rceil}}{S + e^{-\varepsilon \lfloor{u(d(o_1)}\rfloor} + e^{-\varepsilon \lceil{u(d(o_2)}\rceil}} 
    \\&+p_1^\top * p_2^\bot \frac{e^{-\varepsilon \lceil{u(d(o_1)}\rceil}- e^{-\varepsilon \lfloor{u(d(o_2)}\rfloor}}{S + e^{-\varepsilon \lceil{u(d(o_1)}\rceil} + e^{-\varepsilon \lfloor{u(d(o_2)}\rfloor}} 
\end{align*}
Recall that $e^{-\varepsilon u(d,o_1)} \geq e^{e^{-\varepsilon u(d,o_2)}}$ and we only consider the case in which $|u(d,o_1) - u(d,o_2)| \leq 1$. Thus, either $\lfloor{u(d,o_1)}\rfloor = \lfloor{u(d,o_2)}\rfloor$ or $\lfloor{u(d,o_1)}\rfloor +1= \lfloor{u(d,o_2)}\rfloor$.

We now consider each case:
\begin{enumerate}
    \item Suppose $\lfloor{u(d,o_1)}\rfloor = \lfloor{u(d,o_2)}\rfloor = x$. 
    
Given $\lfloor{u(d,o_1)}\rfloor = \lfloor{u(d,o_2)}\rfloor $ and  $e^{-\varepsilon u(d,o_1)} \geq e^{e^{-\varepsilon u(d,o_2)}}$, $u(d,o_1) - \lfloor{u(d,o_1)}\rfloor \leq u(d,o_2) - \lfloor{u(d,o_2)}\rfloor $, and thus $p_1^\top \leq p_2^\top$. 
    
    Then 
    \begin{align*}
    \Pr[\M'(d) = o_1] &- \Pr[\M'(d) = o_2]  \\
    \geq& p_1^\bot * p_2^\top \frac{e^{-\varepsilon x}- e^{-\varepsilon (x+1)}}{S + e^{-\varepsilon x} + e^{-\varepsilon (x+1)}} 
    \\&+p_1^\top * p_2^\bot \frac{e^{-\varepsilon (x+1)}- e^{-\varepsilon x}}{S + e^{-\varepsilon x} + e^{-\varepsilon(x+1)}} \\
     =& v(p_1^\bot * p_2^\top  - p_1^\top * p_2^\bot )\\
    =& v(p_2^\top - p_1^\top)\\
    \geq& 0
\end{align*}
where $v = \frac{e^{-\varepsilon x}- e^{-\varepsilon (x+1)}}{S + e^{-\varepsilon x} + e^{-\varepsilon (x+1)}} > 0$.

    \item Suppose $\lfloor{u(d,o_1)}\rfloor +1= \lfloor{u(d,o_2)}\rfloor = x+1$.
    Then
    \begin{align*}
     \Pr[\M'(d) = o_1] &- \Pr[\M'(d) = o_2]  \\
    \geq& p_1^\bot * p_2^\top \frac{e^{-\varepsilon x}- e^{-\varepsilon (x+2)}}{S + e^{-\varepsilon x} + e^{-\varepsilon (x+2)}} 
    \\&+p_1^\top * p_2^\bot \frac{e^{-\varepsilon (x+1)}- e^{-\varepsilon (x+1)}}{S + e^{-\varepsilon (x+1)} + e^{-\varepsilon(x+1)}} \\
    =& p_1^\bot * p_2^\top \frac{e^{-\varepsilon x}- e^{-\varepsilon (x+2)}}{S + e^{-\varepsilon x} + e^{-\varepsilon (x+2)}} 
     \\
     \geq& 0
\end{align*}
\end{enumerate}

\end{proof}

We state the bounds more formally in the lemma below:
\begin{lemma}\label{lemma:heuristic}
Given a privacy parameter $\varepsilon$, a utility function $u$ and an outcome set $O$, the difference in the probability of outputting an outcome $o^* \in O$ between the randomized rounding exponential mechanism and the original exponential mechanism is upper-bounded by $A+|B|$ where 

\begin{align}
    p_{u(d,o)} &= |u(d,o) - \lceil{u(d,o)}\rceil|
\\
    q_{u(d,o)} &=  \frac{p_{u(d,o)} e^{-\varepsilon \lfloor{u(d,o)}\rfloor}}{e^{-\varepsilon \lfloor{u(d,o)}\rfloor}+\E[\sum_{o' \in O \backslash \{o\}}e^{-\varepsilon \rho(u(d,o'))}]}+\\ & \frac{(1-p_{u(d,o)}) e^{-\varepsilon \lceil{u(d,o)}\rceil}}{e^{-\varepsilon \lceil{u(d,o)}\rceil}+\E[\sum_{o' \in O \backslash \{o\}}e^{-\varepsilon \rho(u(d,o'))}]}
    \\
    A &= 1-\sum_{o \in O}q_{u(d,o)}  \\
    B&= \frac{e^{-\varepsilon u(d,o^*)}}{\sum_{o \in O}e^{-\varepsilon u(d,o)}} - q_{u(d,o^*)}
\end{align}

\end{lemma}
\noindent 
A complete proof of the Lemma is included in the Appendix. To build intuition, we give the following informal explanation of each part.
\begin{itemize}
    \item [$p_{u(d,o)}$] is the probability that $u(d,o)$ is rounded down.
    \item [$q_{u(d,o)}$] is a lower bound on the probability that the randomized rounding mechanism outputs $o$. The first term is a lower bound on the probability conditioned on rounding down, and the second is conditioned on rounding up. 
    \item [$A$] is the cumulative estimation error of the lower bound for all elements and therefore an upper bound on the estimation error for $q_{u(d,o)}$ for any single element. Why? Each element's probability is lower bounded by $q_{u(d,o)}$, so the sum $\sum_{o\in O}q_{u(d,o)}$ would be 1 if the estimates were exact. The maximum error on any given element's probability is therefore given by subtracting the sum of the lower bounds from 1. 
    \item [$B$] is the difference between the lower bound probability for the randomized rounding mechanism and the probability given by the original mechanism for outputting $o^*$. 
    \item[$A+|B|$] is the sum of the error on the lower bound and the difference between the lower bound and the original unrounded mechanism probability. It is therefore the maximum difference between the probabilities of outputting $o^*$ in the original unrounded mechanism and the randomized rounding mechanism.
\end{itemize}

\noindent The tightness of the error bound is directly related to the size of $A$, and as we  see in Section \ref{section:laplace}, the Lemma gives a tight estimate of the error in the important setting of implementing Laplace noise.

\paragraph{Proof of Lemma \ref{lemma:heuristic}}
\begin{proof}


We break the proof of the lemma into several parts for easier reading. In part 1, we will show that $q_{u(d,o)}$ is indeed a lower bound for the probability that the randomized rounding exponential mechanism outputs $o$. In part 2, we will show that $A+|B|$ is an upper bound on the difference in probability that $o^*$ is output by the randomized rounding exponential mechanism versus the original exponential mechanism.

\paragraph{Part 1.} Given a constant $a$ and a set of random variables $X_i$, notice that $\E[\frac{a}{a + \sum_i X_i}] \geq \frac{a}{a + \sum_i \E[X_i]}.$ This follows from observing that for a random variable $Y$, $\E[Y] \leq \frac{1}{\E[\frac{1}{Y}]}$ from Jensen's inequality, and that $\E[a*Y] = a*\E[Y]$ and $\E[a+Y] = a + \E[Y]$ and $\E[Y_1 + Y_2] = \E[Y_1] + \E[Y_2]$ for independent random variables $Y_1$ and $Y_2$ by linearity of expectation. 

The next step is to translate these observations to our scenario. Consider the expression for the expected value of the probability of selecting $o$ in the randomized rounding exponential mechanism:

\[\Pr[\M(u,d) =o ] = \E[\frac{e^{-\varepsilon \rho(u(d,o))}}{\sum_{o \in O}e^{-\varepsilon \rho(u(d,o))}}]\]

Notice that, as written, we cannot exactly translate the above using our previous observations as $e^{\varepsilon \rho((u(d,o))}$ is not a constant - it is a random variable. To get around this, we condition on $u(d,o)$ being rounded either up or down:
\begin{align*}
    \Pr[\M(u,d) =o ] =& p_{u(d,o)}\Pr[\M(u,d) =o |\text{ round down}] +\\ & (1-p_{u(d,o)}) \Pr[\M(u,d) =o |\text{ round up}] \\
    =& p_{u(d,o)}\E[\frac{e^{-\varepsilon \lfloor{u(d,o))}\rfloor}}{e^{-\varepsilon \lfloor{u(d,o))}\rfloor} + \sum_{o' \in O \backslash \{o\}}e^{-\varepsilon \rho(u(d,o'))}}] + \\ & (1-p_{u(d,o)})\E[\frac{e^{-\varepsilon \lceil{u(d,o))}\rceil}}{e^{-\varepsilon \lceil{u(d,o))}\rceil} + \sum_{o' \in O \backslash \{o\}}e^{-\varepsilon \rho(u(d,o'))}}]
\end{align*}

Now we have two terms with constant numerators, identical to $\E[\frac{a}{a + \sum_i X_i}] $ where \\ $X_i = e^{-\varepsilon \rho(u(d,o_i))}$ and $a =e^{-\varepsilon \lceil{u(d,o))}\rceil} $ or $e^{-\varepsilon \lfloor{u(d,o))}\rfloor}$.
Thus,
\begin{align*}
    \Pr[\M(u,d) =o ] \geq & \frac{p_{u(d,o)}e^{-\varepsilon \lfloor{u(d,o))}\rfloor}}{e^{-\varepsilon \lfloor{u(d,o))}\rfloor} + \sum_{o' \in O \backslash \{o\}}\E[e^{-\varepsilon \rho(u(d,o'))}]} + \\ & \frac{(1-p_{u(d,o)})e^{-\varepsilon \lceil{u(d,o))}\rceil}}{e^{-\varepsilon \lceil{u(d,o))}\rceil} + \sum_{o' \in O \backslash \{o\}}\E[e^{-\varepsilon \rho(u(d,o'))}]} \\
    \Pr[\M(u,d) =o ] \geq& q_{u(d,o)}
\end{align*}

\paragraph{Part 2.} Call the probability that the randomized rounding mechanism outputs $o$ $s_o$. The difference in probability on any given element of the randomized rounding mechanism versus the original mechanism is $|s_o -\frac{e^{-\varepsilon u(d,o)}}{\sum_{o \in O}e^{-\varepsilon u(d,o)}} |$. Taking the lower bound $q_{u(d,o)}$ as an approximation for $s_o$, notice that $s_o - q_{u(d,o)} \leq 1 - \sum_{o \in O} q_{u(d,o)}$. This follows from observing that $q_{u(d,o)}$ is a lower bound for $s_o$ for all $o \in O$, and thus the sum of the lower bounds is a lower bound on the sum, which is known to be 1. Thus the maximum error $q_{u(d,o)}$ may have on any $o \in O$ is $A = 1 - \sum_{o \in O} q_{u(d,o)}$. Using $q_{u(d,o)}$ as an approximation, the maximum error on $\max\{ A -B, B\} \leq A + | B|$.

\end{proof}

Figures \ref{fig:pointwise} and \ref{fig:boundcomparison} illustrate the consequences of Lemma \ref{lemma:heuristic} for the Laplace application. Figure \ref{fig:pointwise} shows that the maximum point-wise error between the randomized rounding and unrounded mechanisms is small, and decreases as the granularity $\gamma$ of the outcome space decreases. Intuitively, the smaller $\gamma$ is and the more outcomes are included, the closer $\E[\frac{1}{X}]$ should approximate $\frac{1}{\E[X]}$, as $X$ concentrates with more samples. Figure \ref{fig:boundcomparison} illustrates how the lower and upper bounds compare to the naive probabilities, and also shows that the lower bound error is small. How should we interpret these results? Loosely speaking, applying the lemma to a particular application gives us a worst-case bound on how different the probabilities of the randomized rounding mechanism will be from the unrounded mechanism. As we see in the case of Laplace, these errors are not too large, even in the worst case.

\begin{figure}%
\centering
\includegraphics[width=0.5\textwidth]{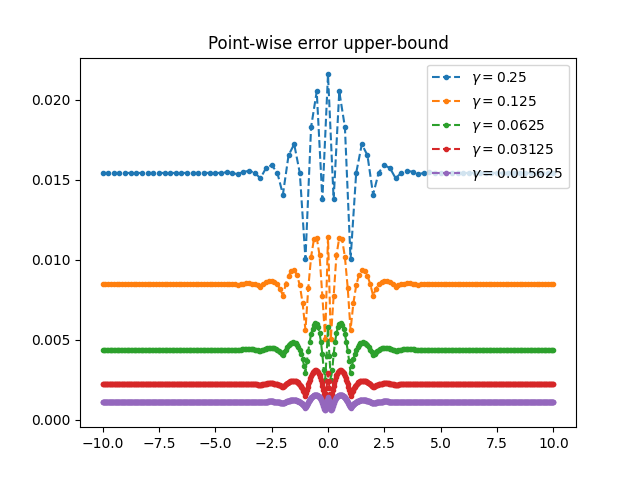}
\caption{The maximum possible point-wise difference between the output probabilities of the unrounded naive mechanism and the randomized rounding exponential mechanism at varying granularities. Notice that as granularity increases, the error decreases.}
\label{fig:pointwise}
\end{figure}

\begin{figure}%
\centering
\includegraphics[width=0.5\textwidth]{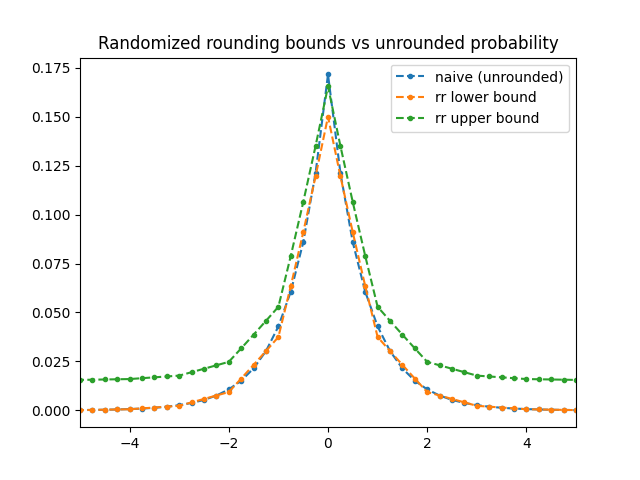}%
\caption{Accuracy bounds for randomized rounding versus the probability assigned by the naive mechanism.}
\label{fig:boundcomparison}
\end{figure}

\subsubsection*{Normalized Sampling}
Proof of Lemma \ref{lemma:normalizesampling} for Algorithm \ref{alg:psweightedsample}.
\begin{proof}
The proof consists of three parts: first, correctness of the distribution assuming sufficient precision; second, sufficiency of precision and bounding the number of random bits required; third, identification of error conditions.

\textit{Correctness.} Notice that (assuming infinite precision) this procedure amounts to partitioning the range $[0,t)$, where $t$ is the total weight, between the elements of $W$ according to their weight, sampling a value $s$ in $[0,t)$ and choosing an element in $W$ based on which partition $s$ lands in. 
To see the correctness of the sampling procedure, observe that in each iteration of the  \textbf{while} loop, any elements ``ruled out'' by the bits of $s$ we have seen so far are removed. That is, any element that has been assigned a range strictly smaller than $s$ or strictly larger than $s+2^{j-1}$ is removed. This leaves the set of elements which still could be reached by the remaining bits of $s$, namely elements that have some portion of the range $[s,s+2^{j})$. The probability that any given element is sampled is equivalent to the probability that a random value $s \in [0,t)$ falls into its assigned range of $[0,t)$, thus, each element is sampled with probability $\frac{w_i}{t}$, which is equivalent to the exponential mechanism. 


\textit{Randomness required.} Notice that each iteration of the \textbf{while} loop at Line 23 requires $p$ random bits. 
Thus it suffices to show that the loop only resets (i.e., the $s$ is rejected) a constant number of times with high probability. 
Notice that $s$ can only be greater than or equal to $t$ if the first random value selected is $1$, as $t$ requires at least $g$ bits to write. 
\footnote{Strictly speaking, $s$ will also need at least one additional lower bit to be $1$, but we give the simpler version as it's sufficient for our purposes.} The probability that $k$ samples all have $1$ as their first bit is $2^{-k}$. Thus, with probability at least $1-2^{-k}$, $k$ iterations will result in a successful sample.

\textit{Sufficient precision.} Notice that $p$ bits 
are sufficient to express any $c_i$ for $i \in [|W|]$. Imagine that an oracle agrees to read out a random value in $[0,t)$ with infinite bits of precision. After hearing $p$ bits, we have sufficient information to choose a single value in $W$, and hearing any more bits cannot change our choice. The sum in Line \ref{line:sum} is equivalent to taking the first $p$ bits from the oracle and writing them as a value in $[0,1)$ scaled by $2^g$. 
This follows from observing that at most one element can ``claim'' any range $[a2^{-p},(a+1)2^{-p})$ as all  combinations of $W$ can be expressed in $p$ bits of precision. 
Thus, $p$ bits of precision are sufficient to simulate the infinite precision procedure.

\textit{Error conditions.}
Suppose again that the procedure has access to infinite bits of randomness (i.e., that the procedure can continue by sampling additional randomness if $j<0$ in Line 17). In this case, the procedure correctly samples from the distribution. Now, suppose that sufficient precision is not available. There are two possible cases, either (1) the value sampled covered a region containing $[s, s + s^*_{g-p}2^{g-p})$, and hearing any more bits of $s^*$ wouldn't change the outcome or (2) an error is returned. Thus, either the sample is the sample that would have been drawn given infinite precision, or an error is returned. 

\end{proof}

\subsubsection*{Optimized Sampling}
\begin{algorithm}[htbp]
  \caption{The optimized normalized weighted sampling algorithm}
  \label{alg:optimizedweightedsample}
\begin{algorithmic}[1]
  \LineComment{\textbf{Inputs}: $W$, a set of weights.} 
  \LineComment{\textbf{Outputs}: $i$, an index sampled according to $p(i) := \frac{w_i}{\sum_{w_j \in W} w_j}$.}
    \Procedure{NormalizedSample}{$W$}
    \State $t \leftarrow \sum_{w \in W}w$ \Comment{The total weight.}
    \For{$i \in \{1,\ldots,|W|\}$} \Comment{Compute the cumulative weights.}
      \State $c_i \leftarrow \sum_{j=1}^i w_j$  \Comment{Each element assigned $[c_{i-1},c_i)$.}
    \EndFor
    \State $k \leftarrow \argmin_k \{2^k \geq t\}$
    \Comment{The smallest power of two $\geq t$.}
    \If{$2^k>t$}
    \State $W \leftarrow W \cup \{\bot\}$ \Comment{Add element $\bot$ with weight $2^k-t$.}
    \State $c_{|W|} \leftarrow 2^k-t$\Comment{Total weight is now $2^k$.}
    \EndIf
    \State $s \leftarrow 0$
    \State $j \leftarrow k -1$
    \State $R \leftarrow [|W|]$ \Comment{the remaining elements}
    \While{$|R| > 1$}
        \State $r_j \sim \mathbf{Unif}(0,1)$
        \State $s \leftarrow s + r_j 2^j$
        \For{$i \in R$}
            \If{$c_i \leq s $} \Comment{$s$ cannot be in $[c_{i-1},c_i)$, even if all draws are $0$.}
                \State $R \leftarrow R \backslash \{i\}$ 
            \EndIf
            \If{$i>0$ and $c_{i-1} \geq s + 2^{j} $} \Comment{$s$ cannot be in $[c_{i-1},c_i)$, even if maximum value added.}
                \State $R \leftarrow R \backslash \{i\}$ 
            \EndIf
        \EndFor
        \State $j \leftarrow j - 1$
        \If{$|R| = 1$ and $\bot \in R$} \Comment{Restart if dummy value.}
            \State $s \leftarrow 0$
            \State $j \leftarrow k-1$
            \State $R \leftarrow [|W|]$
        \EndIf
    \EndWhile
    
    \State \textbf{return} $l$ 
    \EndProcedure
  \end{algorithmic}
\end{algorithm}

\end{document}